\documentclass[aps,prd,longbibliography,nofootinbib,floatfix]{revtex4-2}

\usepackage{graphicx}
\usepackage{xcolor}
\usepackage{amsmath,amssymb,bm,mathtools}
\usepackage{siunitx}
\usepackage{booktabs}
\usepackage{microtype}
\usepackage{natbib}
\usepackage{hyperref}
\usepackage{threeparttable}
\hypersetup{colorlinks=true,allcolors=blue}
\usepackage{comment}
\newcommand{\lcdm}{\ensuremath{\Lambda\mathrm{CDM}}}
\newcommand{\nulcdm}{\ensuremath{\nu\Lambda\mathrm{CDM}}}

\newcommand{\wcdm}{\ensuremath{w\mathrm{CDM}}}
\newcommand{\wacdm}{\ensuremath{w_0w_a\mathrm{CDM}}}
\newcommand{\rd}{r_{\rm d}}

\newcommand{\DM}{D_{\rm M}}
\newcommand{\DHu}{D_{\rm H}}

\newcommand{\KL}{D_{\rm KL}}
\newcommand{\Lth}{\mathcal{L}}

\newcommand{\Ltot}{\mathcal{L}_{\rm tot}}
\newcommand{\vect}[1]{\boldsymbol{#1}}

\newcommand{\tablefont}{\small}
\newcommand{\tighttable}{%
  \setlength{\tabcolsep}{4.5pt}%
  \renewcommand{\arraystretch}{1.16}%
}
\newcommand{\tcell}[2]{\begin{minipage}[t]{#1}\raggedright\strut #2\strut\end{minipage}}

\usepackage{pgfplots}
\usepgfplotslibrary{groupplots,fillbetween}
\usetikzlibrary{calc,positioning,arrows.meta}
\pgfplotsset{compat=1.18}

\pgfplotsset{
  desiAxis/.style={
    tick align=outside,
    tick style={black, line width=0.45pt},
    axis line style={line width=0.45pt},
    xlabel style={font=\small},
    ylabel style={font=\small},
    tick label style={font=\footnotesize},
    legend style={
      font=\footnotesize,
      draw=none,
      fill=none,
      cells={anchor=west}
    },
    legend cell align=left,
    grid=major,
    major grid style={draw=black!10, line width=0.20pt},
    samples=201,
    clip mode=individual
  },
  theorycurve/.style={black, line width=0.50pt},
  refcurve/.style={black, dashed, line width=0.50pt},
  datapts/.style={only marks, mark=*, mark size=1.05pt, black}
}

\usepackage{amsthm}

\newtheorem{proposition}{Proposition}

\theoremstyle{remark}
\newtheorem{remark}{Remark}

\newcommand{\FAR}{\mathrm{FAR}}
\newcommand{\ECE}{\mathrm{ECE}}
\newcommand{\NLS}{\mathrm{NLS}}
\newcommand{\JS}{\mathrm{JS}}
\newcommand{\one}{\mathbf{1}}

\newcommand{\Act}{\mathbf A}
\newcommand{\ActSpace}{\mathcal A}
\newcommand{\Vspace}{\mathcal V}
\newcommand{\Cat}{\mathcal M}
\newcommand{\FAP}{F_{\rm AP}}
\newcommand{\DV}{D_{\rm V}}

\newcommand{\logit}{\operatorname{logit}}
\newcommand{\BS}{\mathrm{BS}}
\newcommand{\TV}{\mathrm{TV}}

\allowdisplaybreaks[1]
\begin{document}

\title{Sector-Resolved Bayesian Model Averaging for DESI-Era Cosmology}

\author{Slava G. Turyshev}
\affiliation{%
Jet Propulsion Laboratory, California Institute of Technology,\\
4800 Oak Grove Drive, Pasadena, CA 91109-0899, USA
}%

\date{\today}

\begin{abstract}
We present a quotient-space Bayesian formulation for DESI-era anomaly interpretation. Given a pattern-labeled catalog with map \(i\mapsto \Act(i)\), the induced posterior \(p(\Act\mid D)\), sector inclusion probabilities \(P_\alpha\), co-activation probabilities \(P_{\alpha\beta}\), and grouped Bayes factors \(B_\alpha(D)\) are exact summaries over predeclared physical activation events. Pairwise comparisons such as \(\lcdm\) versus \(\wacdm\) remain ordinary Bayes-factor tests between specified families; the quotient construction addresses the coarser question of which physical sector carries posterior support when different sectors are represented by unequal numbers of catalog elements. We derive a sector-resolved DESI--CMB--SN likelihood specification for late-time background, early-time ruler, supernova calibration, perturbation, and gravitational-wave propagation sectors. The construction includes an Alcock--Paczynski/isotropic-scale BAO decomposition, a pure-ruler projection, analytic marginalization of low-rank supernova calibration modes, Fisher-normalized sector priors, inactive-sector leakage tests, log-evidence uncertainty propagation, and prior/sector-partition diagnostics. The result is a quantitative procedure for reporting model-comparison support at the level of physically interpretable sectors.
\end{abstract}

\maketitle

\section{Introduction}
\label{sec:intro}

Dark Energy Spectroscopic Instrument (DESI) era cosmology is beginning to face questions that are not exhausted by individual family-by-family model comparisons. The issue is no longer only whether a posterior in the \((w_0,w_a)\) plane \cite{Chevallier2001,Linder2003} drifts away from \(w=-1\). It is whether the available probes can distinguish among physically different sources of an apparent anomaly: late-time background deformation, early-time ruler physics, low-rank supernova calibration structure, perturbation-sector freedom, or tensor-sector propagation effects.

This distinction is already important for current DESI analyses. Official DESI Data Release 2 (DR2) analyses report support for late-time extensions of \(\Lambda\) cold dark matter (\(\Lambda\)CDM) in baryon acoustic oscillation (BAO)-based
combinations with the cosmic microwave background (CMB), whereas recent Bayesian reanalyses find that the corresponding evidence can weaken substantially once
prior volume is integrated over \cite{DESIDR2BAO2025,DESIDR2Extended2025,
OngYallupHandley2025,OngYallupHandley2026}. In parallel, recent supernova recalibration studies suggest that part of the apparent support for evolving dark energy can be absorbed by low-rank calibration structure \cite{DESDovekie2025,HuangCaiWang2025}. Posterior contours, Bayes factors, and data-set consistency metrics therefore answer different questions and should be reported as distinct statistical summaries.

Pairwise model comparisons and sector-level inference answer different Bayesian questions. A pre-specified comparison between two well-defined families, for example \(M_0=\lcdm\) and \(M_1=\wacdm\), is summarized by the usual Bayes factor \(B_{10}=Z_1/Z_0\). A sector claim, such as activation of the late-time background sector, is a coarser event obtained by summing posterior support over all catalog elements that realize that event. The quotient construction defines this event-level posterior before evidence aggregation, so that the answer is insensitive to prior-preserving refinements inside a fixed activation pattern.

The corresponding sector-level analysis treats the latent sector activated by the data as the inferential target, with phenomenological families serving as realizations of sector hypotheses. BAO measure both an Alcock--Paczynski shape
variable and a scale variable involving \(\DV/\rd\), and hence constrain combinations of \(E(z)\), distance ratios, and
\(\rd\), not \(H_0\) and \(\rd\) separately. Supernovae constrain relative luminosity distances only after calibration and selection structure have been marginalized. Full-shape clustering, redshift-space distortions, weak lensing, and CMB lensing probe scalar perturbations rather than the homogeneous
background alone. Standard sirens probe tensor-sector propagation. A claimed anomaly should therefore be localized at sector level before being interpreted as a particular microphysical model.

Bayesian product-space and trans-dimensional methods for model comparison are well established \cite{CarlinChib1995,Green1995,Trotta2008}, and Bayesian model averaging with posterior inclusion probabilities is standard statistical technology \cite{Hoeting1999BMA}. Bayesian model dimensionality and suspiciousness-based tension diagnostics are likewise established in cosmology
\cite{HandleyLemos2019,Lemos2019Suspiciousness,Hergt2026Consistency}. This work defines a DESI-specific inferential target: posterior inference on sector-activation patterns. The resulting quotient-space construction yields exact pattern posteriors, sector inclusion probabilities, pairwise co-activation probabilities, and grouped sector Bayes factors that depend on the catalog only through grouped evidences and pattern priors. In the DESI setting, this is the relevant invariance: sector conclusions should not change because one theoretical
sector is represented by more catalog elements than another.

The paper has three aims. First, it defines the exact sector-level objects to be reported: \(p(\Act\mid D)\), \(P_\alpha\), \(P_{\alpha\beta}\), and \(B_\alpha(D)\). Second, it derives a minimal sector-resolved likelihood specification and Fisher-normalized detection basis for late-time background,
ruler, supernova (SN)-calibration, perturbation, and gravitational-wave (GW)-propagation sectors. Third, it specifies the numerical-precision, data-block attribution, pattern-prior, sector-partition, calibration, posterior-predictive, and
robustness criteria required before a DESI-era anomaly can be interpreted physically. The recommended reporting set consists of pattern posteriors, sector posterior summaries, grouped Bayes factors, data-block sensitivity checks, BAO/SN diagnostic residuals, and calibration and robustness diagnostics.

The paper is organized as follows. Section~\ref{sec:stats} defines the inferential target by separating posterior estimation, evidence-based model comparison, and data-set consistency. Section~\ref{sec:grouped_likelihood} develops the grouped estimator and the associated likelihood specification. Section~\ref{sec:validation} states the validation criteria for sector-level claims. Section~\ref{sec:concl} summarizes the physical scope of the construction and its role in DESI-era cosmology.

\section{Inferential target after DESI}
\label{sec:stats}

Our objective is to infer which latent sector is activated by the
data, rather than merely which phenomenological fit minimizes residuals. In the present context, the relevant sector decomposition is
\begin{equation}
\Theta
=
\Theta_{\rm base}
\cup
\Theta_{E(z)}
\cup
\Theta_{\rd}
\cup
\Theta_{\rm SN}
\cup
\Theta_{\rm pert}
\cup
\Theta_{\rm GW},
\label{eq:sectordecomp}
\end{equation}
where \(\Theta_{\rm base}\) denotes the baseline cosmology,
\(\Theta_{E(z)}\) the late-time expansion-shape sector,
\(\Theta_{\rd}\) the early-time ruler sector,
\(\Theta_{\rm SN}\) the low-rank supernova calibration/selection sector,
\(\Theta_{\rm pert}\) perturbation-level freedom in growth and lensing, and
\(\Theta_{\rm GW}\) the gravitational-wave (GW) tensor sector.

Three inferential tasks must be separated. The first is posterior estimation within a fixed model \(M\):
\begin{equation}
p(\theta\mid D,M)
=
\frac{\Lth(D\mid\theta,M)\,\pi(\theta\mid M)}{Z_M},
\qquad
Z_M
=
\int d\theta\,
\Lth(D\mid\theta,M)\,\pi(\theta\mid M).
\label{eq:posterior}
\end{equation}
This answers the question: what parameter values are allowed if model \(M\) is assumed?

The second task is evidence-based model comparison. A useful identity is
\begin{equation}
\ln Z_M
=
\big\langle \ln\Lth \big\rangle_{p(\theta\mid D,M)}-\KL,
\label{eq:lnZdecomp}
\end{equation}
which makes explicit why a model can improve the best fit yet lose in evidence: the gain in fit quality need not compensate for the required compression of prior volume \cite{Hergt2026Consistency}.  In Eq.~(\ref{eq:lnZdecomp}), \(D_{\rm KL}\) denotes the Kullback--Leibler (KL) divergence between posterior and prior within model \(M\). A related diagnostic is the Bayesian model dimensionality,
\begin{equation}
d
\equiv
2\,{\rm Var}_{p(\theta\mid D,M)}
\!\big[\ln \Lth(D\mid\theta,M)\big]
=
2\,{\rm Var}_{p(\theta\mid D,M)}
\!\left[
\ln\frac{p(\theta\mid D,M)}{\pi(\theta\mid M)}
\right],
\label{eq:dimensionality}
\end{equation}
which measures the number of effectively constrained directions rather than the raw parameter count \cite{HandleyLemos2019}. For an \(n\)-dimensional Gaussian posterior fully contained in the prior support, Eq.~(\ref{eq:dimensionality}) recovers \(d=n\).

The third task is data-set consistency. If \(A\) and \(B\) are two data blocks, the evidence ratio and suspiciousness are
\begin{equation}
\ln R_{AB}=\ln Z_{AB}-\ln Z_A-\ln Z_B,
\qquad
\ln S=\ln R-\ln I,
\label{eq:suspiciousness}
\end{equation}
with \(\ln I\) the information-gain correction constructed from the corresponding Kullback--Leibler divergences \cite{Hergt2026Consistency}. In approximately Gaussian cases, \(-2\ln S\) may be calibrated using the effective shared dimensionality \(d_{\rm sh}=d_A+d_B-d_{AB}\). Posterior shifts and posterior-predictive checks then localize where the disagreement lives.

This separation is central to DESI-era interpretation. The official \(\sim 3.1\sigma\) preference for \wacdm\ in BAO+CMB \cite{DESIDR2BAO2025} and the mild Bayesian preference for \lcdm\ in DESI DR2+Planck \cite{OngYallupHandley2026} address different inferential tasks. These tasks are complementary rather than interchangeable. Posterior estimation constrains parameters within a fixed model; evidence-based comparison tests whether an extension is warranted; and data-set consistency diagnostics determine whether the contributing data blocks are mutually compatible. A complete post-DESI analysis should report all three. This three-way separation fixes the inferential target. The role of the grouped construction developed in Sec.~\ref{sec:grouped_likelihood} is to lift model comparison from individual model labels to sector-activation events without discarding ordinary posterior inference within models.

Table~\ref{tab:notation_guide} provides a compact reference for notation, units, and acronyms used throughout the paper. All logarithms are natural logarithms; logarithmic distances and scale variables are dimensionless, while dimensional distances are measured in the units adopted by the input likelihood, usually Mpc.

\begin{table}[!tbp]
\caption{Notation, symbols, units, and acronyms used in the manuscript. The table is intended as a quick reference for the grouped-inference layer and for the DESI--CMB--SN likelihood diagnostics.}
\label{tab:notation_guide}
\tablefont
\tighttable
\begin{tabular}{@{}lll@{}}
\toprule
\tcell{0.20\textwidth}{Symbol or acronym} &
\tcell{0.52\textwidth}{Definition} &
\tcell{0.20\textwidth}{Units or convention} \\
\midrule
\tcell{0.20\textwidth}{\(D\), \(D_b\), \(D_{\setminus b}\)} &
\tcell{0.52\textwidth}{Full compressed data vector, data block \(b\), and data vector with block \(b\) removed.} &
\tcell{0.20\textwidth}{Data-vector units} \\
\tcell{0.20\textwidth}{\(i\), \(\Cat\), \(\Cat_{\Act}\)} &
\tcell{0.52\textwidth}{Neutral catalog index, full catalog, and catalog elements assigned to activation pattern \(\Act\).} &
\tcell{0.20\textwidth}{Discrete labels} \\
\tcell{0.20\textwidth}{\(\Vspace_i\), \(\vartheta_i\)} &
\tcell{0.52\textwidth}{Parameter space and parameters of catalog element \(i\).} &
\tcell{0.20\textwidth}{Model dependent} \\
\tcell{0.20\textwidth}{\(\Act\), \(\ActSpace\), \(A_\alpha\)} &
\tcell{0.52\textwidth}{Binary activation vector, allowed pattern space, and indicator for sector \(\alpha\).} &
\tcell{0.20\textwidth}{\(A_\alpha\in\{0,1\}\)} \\
\tcell{0.20\textwidth}{\(\Theta_{\rm base}\), \(\Theta_{E(z)}\), \(\Theta_{\rd}\), \(\Theta_{\rm SN}\), \(\Theta_{\rm pert}\), \(\Theta_{\rm GW}\)} &
\tcell{0.52\textwidth}{Baseline, late-time background, early-time ruler, supernova calibration, perturbation, and tensor-propagation sectors.} &
\tcell{0.20\textwidth}{Physical sectors} \\
\tcell{0.20\textwidth}{\(Z_i\), \(Z_{\Act}\), \(\ln B_\alpha\)} &
\tcell{0.52\textwidth}{Catalog-element evidence, grouped pattern evidence, and sector inclusion Bayes factor.} &
\tcell{0.20\textwidth}{Evidence dimensionless; logs in nats} \\
\tcell{0.20\textwidth}{\(\pi(\Act)\), \(\pi(i\mid\Act)\), \(q_\alpha\)} &
\tcell{0.52\textwidth}{Pattern prior, within-pattern catalog prior, and marginal prior probability that sector \(\alpha\) is active.} &
\tcell{0.20\textwidth}{Probabilities} \\
\tcell{0.20\textwidth}{\(P_\alpha\), \(P_{\alpha\beta}\)} &
\tcell{0.52\textwidth}{Posterior sector activation and pairwise co-activation probabilities.} &
\tcell{0.20\textwidth}{Probabilities} \\
\tcell{0.20\textwidth}{\(E(z)\), \(\rd\), \(\DM\), \(\DHu\), \(\DV\), \(\FAP\)} &
\tcell{0.52\textwidth}{Normalized expansion rate, sound horizon, transverse comoving distance, Hubble distance, isotropic BAO distance, and Alcock--Paczynski variable.} &
\tcell{0.20\textwidth}{Distances in likelihood units; ratios dimensionless} \\
\tcell{0.20\textwidth}{\(\Delta\vect y\), \(\Delta\vect v\), \(a_{\rd}\), \(\Delta\vect v_\perp\)} &
\tcell{0.52\textwidth}{BAO AP residual, isotropic-scale residual, pure-ruler projection, and scale residual orthogonal to the pure-ruler direction.} &
\tcell{0.20\textwidth}{Dimensionless logarithmic residuals} \\
\tcell{0.20\textwidth}{\(M\), \(\vect q\), \(\widehat{\vect q}\), \(\vect r_\mu^\perp\)} &
\tcell{0.52\textwidth}{SN calibration design matrix, calibration coefficients, posterior mean calibration mode, and calibration-orthogonal residual.} &
\tcell{0.20\textwidth}{Magnitudes for SN residuals} \\
\tcell{0.20\textwidth}{\(\vect s_\alpha\), \(F_\alpha\), \(\sigma_\alpha\)} &
\tcell{0.52\textwidth}{Fisher-normalized sector amplitudes, Fisher metric, and active-sector prior width.} &
\tcell{0.20\textwidth}{Dimensionless in Fisher units} \\
\tcell{0.20\textwidth}{\(\FAR\), \(\ECE\), PPC, \(\JS\), \(\TV\)} &
\tcell{0.52\textwidth}{False-activation rate, expected calibration error, posterior-predictive check, Jensen--Shannon divergence, and total variation distance.} &
\tcell{0.20\textwidth}{Probabilities except \(\JS\) in nats} \\
\tcell{0.20\textwidth}{DESI, BAO, CMB, SN, FS, RSD, WL, GW} &
\tcell{0.52\textwidth}{Dark Energy Spectroscopic Instrument, baryon acoustic oscillations, cosmic microwave background, supernovae, full-shape clustering, redshift-space distortions, weak lensing, and gravitational waves.} &
\tcell{0.20\textwidth}{Acronyms} \\
\bottomrule
\end{tabular}
\end{table}

\section{Grouped estimator and sector-resolved likelihood specification}
\label{sec:grouped_likelihood}

This section proceeds from exact grouped inference on a pattern-labeled catalog to the sector parameterization and then to the likelihood specification. Section~\ref{subsec:grouped_estimator} defines the grouped estimator. Section~\ref{subsec:quotient_invariance} gives its quotient-space interpretation and refinement invariance. Section~\ref{subsec:basis_priors} introduces the Fisher-normalized detection basis and the associated standardized sector priors. Section~\ref{subsec:likelihood} specifies the minimal likelihood model used for sector-resolved inference. Section~\ref{subsec:implementation} records the numerical and methodological requirements for applying the construction in a likelihood analysis. Section~\ref{subsec:geometry_application} specifies a minimal eight-pattern DESI--CMB--SN geometry analysis.

\subsection{Grouped estimator on a pattern-labeled catalog}
\label{subsec:grouped_estimator}

For rigor, the catalog is treated as a disjoint union of model-specific
parameter spaces rather than as a single common Euclidean space:
\begin{equation}
\Vspace \equiv
\bigsqcup_{\Act\in\ActSpace}\bigsqcup_{i\in\Cat_{\Act}}\Vspace_i,
\qquad
\vartheta_i\in\Vspace_i.
\label{eq:disjoint_union}
\end{equation}
Here \(i\) is a neutral catalog index, \(\Vspace_i\) is the parameter space of that catalog element, \(\ActSpace\) is the set of allowed binary activation patterns, and \(\Cat_{\Act}\) is the set of catalog elements assigned to pattern \(\Act\). This notation is deliberately distinct from the standard cosmological
symbols \(\Omega_m\), \(h\), and \(w\).

The hypermodel posterior is
\begin{equation}
p(\Act,i,\vartheta_i\mid D)
\propto
\one\!\{i\in\Cat_{\Act}\}\,
\pi(\Act)\,\pi(i\mid \Act)\,\pi(\vartheta_i\mid i)\,
\Lth(D\mid \vartheta_i,i),
\label{eq:hypermodel_rigorous}
\end{equation}
with catalog-element evidence
\begin{equation}
Z_i
=
\int_{\Vspace_i} d\vartheta_i\,
\Lth(D\mid \vartheta_i,i)\,\pi(\vartheta_i\mid i).
\label{eq:Zi_rigorous}
\end{equation}

The activation pattern is the binary vector
\begin{equation}
\Act
\equiv
\big(A_{E(z)},A_{\rd},A_{\rm SN},A_{\rm pert},A_{\rm GW}\big),
\qquad
A_\alpha(\Act)\in\{0,1\},
\label{eq:pattern}
\end{equation}
where the baseline sector is understood to be always active. Every catalog element \(i\) is assigned to exactly one pattern \(\Act(i)\). If a broad theoretical family can realize more than one activation pattern over its parameter support, it must be split into pattern-labeled catalog elements before grouped inference is performed.

The model prior is factorized as
\begin{equation}
\pi(i)=\pi(\Act)\,\pi(i\mid \Act),
\qquad
\sum_{\Act\in\ActSpace}\pi(\Act)=1,
\qquad
\sum_{i\in\Cat_{\Act}}\pi(i\mid \Act)=1.
\label{eq:priorfactorization}
\end{equation}

\paragraph*{Assignment rule for \(\Act(i)\).}
Let \({\cal S}_i:\Vspace_i\rightarrow\prod_\alpha\mathbb R^{N_\alpha}\) denote the map from model parameters to the retained Fisher-normalized sector amplitudes introduced in Sec.~\ref{subsec:basis_priors}. For a candidate catalog element \(i\), define
\begin{equation}
A_\alpha(i)=
\begin{cases}
0, & \vect s_\alpha(\vartheta_i)=\vect 0
     \quad \text{for all }\vartheta_i\in\Vspace_i,\\
1, & \text{otherwise}.
\end{cases}
\label{eq:operational_activation_rule}
\end{equation}
The indicator is structural, not thresholded by detectability. If both values occur on different connected regions of the admissible parameter space, decompose
\begin{equation}
\Vspace_i=\bigsqcup_{\Act\in\ActSpace}\Vspace_{i,\Act},
\qquad
\Vspace_{i,\Act}\equiv
\{\vartheta_i\in\Vspace_i:\Act(\vartheta_i)=\Act\},
\label{eq:pattern_partition}
\end{equation}
and treat each nonempty \(\Vspace_{i,\Act}\) as a distinct catalog element. In this way the pattern map is fixed by the retained detection basis rather than by a verbal source-family label.

\begin{remark}[Structural activation versus detectability]
The indicator \(A_\alpha(i)\) records whether catalog element \(i\) admits a nonzero projection onto sector \(\alpha\) in the retained detection basis on any admissible region of parameter space. Small but allowed amplitudes therefore still correspond to \(A_\alpha=1\); the question of whether the data support activation is delegated to \(P_\alpha\) and \(B_\alpha(D)\).
\end{remark}

The pattern prior must be specified in the main text because posterior activation probabilities depend on it directly. A useful one-parameter family is
\begin{equation}
\pi_\lambda(\Act)
\propto
\pi_0(\Act)\,\exp[-\lambda K(\Act)],
\qquad
K(\Act)=\sum_{\alpha\in\{E,\rd,{\rm SN},{\rm pert},{\rm GW}\}}
A_\alpha(\Act),
\label{eq:patternprior_main}
\end{equation}
where \(\pi_0\) is a stated reference prior on the declared sector partition. The choice \(K=\sum_\alpha A_\alpha\) is not an invariant physical complexity; it is an illustrative sparsity penalty on a declared sector partition. Consequently, the
reporting standard in Sec.~\ref{subsec:sector_partition_prior} requires a pattern-neutral prior, \(\lambda=0\), together with sparsity scans rather than a single privileged default.

Within each activation pattern we recommend
\begin{equation}
\pi(i\mid \Act)
=
\frac{\omega_i}{\sum_{j\in\Cat_{\Act}}\omega_j},
\label{eq:withinpatternprior_main}
\end{equation}
with \(\omega_i=1\) in the absence of intentionally imposed theoretical preferences. The corresponding grouped evidence is
\begin{equation}
Z_{\Act}
\equiv
\sum_{i\in\Cat_{\Act}} Z_i\,\pi(i\mid \Act),
\label{eq:ZAct}
\end{equation}
and posterior mass among activation patterns is
\begin{equation}
p(\Act\mid D)
=
\frac{Z_{\Act}\,\pi(\Act)}
{\sum_{\Act'\in\ActSpace} Z_{\Act'}\,\pi(\Act')}.
\label{eq:patternposterior}
\end{equation}
The posterior for an individual catalog element becomes
\begin{equation}
p(i\mid D)
=
\frac{Z_i\,\pi(i\mid \Act)}{Z_{\Act}}\,p(\Act\mid D),
\qquad
i\in\Cat_{\Act},
\label{eq:modelposterior}
\end{equation}
while model-averaged posteriors follow in the usual way,
\begin{equation}
p(X\mid D)=\sum_{i\in\Cat}p(X\mid D,i)\,p(i\mid D).
\label{eq:modelavg}
\end{equation}

Two distinct sector-level summaries must then be separated. The first is the posterior probability that sector \(\alpha\) is active. The second is the grouped Bayes factor for activation against non-activation:
\begin{equation}
q_\alpha
\equiv
\Pr_\pi(A_\alpha=1)
=
\sum_{\Act:\,A_\alpha(\Act)=1}\pi(\Act),
\qquad
\pi(\Act\mid A_\alpha=a)
=
\frac{\pi(\Act)\,\one\!\{A_\alpha(\Act)=a\}}
{\sum_{\Act':\,A_\alpha(\Act')=a}\pi(\Act')},
\label{eq:sectorpriormass}
\end{equation}
\begin{equation}
B_\alpha(D)
\equiv
\frac{
\sum_{\Act:\,A_\alpha(\Act)=1}
Z_{\Act}\,\pi(\Act\mid A_\alpha=1)
}{
\sum_{\Act:\,A_\alpha(\Act)=0}
Z_{\Act}\,\pi(\Act\mid A_\alpha=0)
},
\label{eq:sectorBF}
\end{equation}
\begin{equation}
P_\alpha
\equiv
p(A_\alpha=1\mid D)
=
\sum_{\Act\in\ActSpace} A_\alpha(\Act)\,p(\Act\mid D),
\qquad
\logit P_\alpha=\ln B_\alpha(D)+\logit q_\alpha.
\label{eq:sectorodds}
\end{equation}
The quantity \(P_\alpha\) is the posterior activation probability under the adopted pattern prior, whereas \(\ln B_\alpha(D)\) isolates the data-driven support for activation after factoring out the marginal prior odds \(q_\alpha/(1-q_\alpha)\). Both should be reported.

Eqs.~(\ref{eq:patternposterior})--(\ref{eq:sectorodds}) define an exact coarse-graining of ordinary Bayesian model averaging from model labels to sector-activation events. At sector level, the resulting summaries are invariant to within-pattern catalog multiplicity while model-level posteriors remain available for projection onto specific theory realizations. The DESI-specific content is the physically motivated quotient map from a cosmological model catalog to predeclared activation events and the associated validation criteria for using that map in current analyses.

\subsection{Quotient-space interpretation and refinement invariance}
\label{subsec:quotient_invariance}

Let \(\sim\) be the equivalence relation on the catalog \(\Cat\) defined by
\begin{equation}
i\sim j
\qquad\Longleftrightarrow\qquad
\Act(i)=\Act(j).
\end{equation}
The pattern space \(\ActSpace\) is then the quotient \(\Cat/\!\sim\). Under the factorized prior \(\pi(i)=\pi(\Act)\pi(i\mid \Act)\), the posterior on the quotient is exactly
\begin{equation}
p(\Act\mid D)=\sum_{i\in\Cat_{\Act}}p(i\mid D),
\label{eq:quotient_bma}
\end{equation}
which follows immediately from Eq.~(\ref{eq:modelposterior}).

The sector activation probability is a posterior inclusion probability for the event \(A_\alpha=1\):
\begin{equation}
P_\alpha
=
\sum_{i\in\Cat}
\one\!\{A_\alpha[\Act(i)]=1\}\,p(i\mid D).
\label{eq:sector_inclusion}
\end{equation}
Likewise, pairwise co-activation probabilities are
\begin{equation}
P_{\alpha\beta}
\equiv
p(A_\alpha=1,A_\beta=1\mid D)
=
\sum_{\Act\in\ActSpace}
A_\alpha(\Act)A_\beta(\Act)\,p(\Act\mid D).
\label{eq:pairwise_activation}
\end{equation}
The quantities \(P_{\alpha\beta}\) should be reported whenever two sectors are observationally degenerate, in particular \((\Theta_{E(z)},\Theta_{\rm SN})\) and \((\Theta_{E(z)},\Theta_{\rd})\). The activation-pattern layer is required
because sector-level conclusions should be insensitive to arbitrary refinements of the catalog.

\begin{proposition}[Prior-preserving refinement invariance]
\label{prop:grouped_invariance}
Let \(r:\widetilde{\Cat}\to\Cat\) be a surjective refinement map such that \(\Act(\widetilde i)=\Act(r(\widetilde i))\). Suppose that for every \(i\in\Cat_{\Act}\),
\begin{equation}
\sum_{\widetilde i\in r^{-1}(i)}
\widetilde\pi(\widetilde i\mid \Act)\,Z_{\widetilde i}
=
\pi(i\mid \Act)\,Z_i.
\label{eq:refinement_map_condition}
\end{equation}
Then all quotient-level posteriors are invariant:
\begin{equation}
\widetilde p(\Act\mid D)=p(\Act\mid D),
\qquad
\widetilde P_\alpha=P_\alpha,
\qquad
\widetilde P_{\alpha\beta}=P_{\alpha\beta},
\qquad
\widetilde B_\alpha(D)=B_\alpha(D).
\label{eq:refinement_map_result}
\end{equation}
\end{proposition}

\begin{proof}[Proof sketch]
Eqs.~(\ref{eq:patternposterior})--(\ref{eq:sectorodds}) depend on the catalog only through the grouped quantities \(\{Z_{\Act},\pi(\Act)\}_{\Act\in\ActSpace}\). If a sector-preserving refinement leaves those grouped quantities unchanged,
then \(p(\Act\mid D)\), \(P_\alpha\), and \(B_\alpha(D)\) are unchanged as well. A one-line derivation is recorded in Appendix~\ref{app:technical}. 
\end{proof}

\begin{remark}[Exactness versus specification dependence]
The grouped identities derived in Eqs.~(\ref{eq:patternposterior})--
(\ref{eq:sectorodds}) are exact conditional on four inputs: the pattern map \(i\mapsto \Act(i)\), the admissible catalog \(\Cat\), the grouped prior structure \(\{\pi(\Act),\pi(i\mid \Act)\}\), and the validated likelihood or compression used to compute the evidences \(Z_i\). The quotient construction
removes arbitrary within-pattern multiplicity and keeps the remaining specification choices explicit: sector partition, basis choice, prior widths, and likelihood approximation. Those quantities are probed by the robustness program
of Sec.~\ref{sec:validation}.
\end{remark}

\subsection{Detection basis and standardized sector priors}
\label{subsec:basis_priors}

To separate sector identification from microphysical interpretation, we parameterize the non-baseline sectors in a low-rank detection basis. The corresponding amplitudes \(s_\alpha\) define the coordinates used both in inference and in robustness studies:
\begin{align}
\delta\ln E(z)
&=
\sum_{a=1}^{N_E} c_a^{E}\,\phi_a^{E}(z),
\qquad
\delta\ln \rd = c_{\rd},
\qquad
\delta\mu_{\rm SN}(z)
=
\Delta m_B + \sum_{a=1}^{N_{\rm SN}} q_a\,\psi_a^{\rm SN}(z),
\label{eq:detbasis1}
\\
\mu(z,k)-1
&=
\sum_{a=1}^{N_\mu} c_a^\mu\,\phi_a^\mu(z,k),
\qquad
\Sigma(z,k)-1
=
\sum_{a=1}^{N_\Sigma} c_a^\Sigma\,\phi_a^\Sigma(z,k),
\qquad
\ln \Xi(z)
=
\sum_{a=1}^{N_\Xi} c_a^\Xi\,\phi_a^\Xi(z).
\label{eq:detbasis2}
\end{align}
The basis functions should be chosen to be orthonormal with respect to the fiducial Fisher metric of the corresponding compressed data block, and the retained rank should satisfy a fixed capture rule. Microphysical model families are then interpreted as constrained submanifolds of this detection basis rather than as the primary detection space itself.

The data sensitivity is organized as follows. BAO shape carries the leading discrimination power for \(\Theta_{E(z)}\), the sound-horizon amplitude for \(\Theta_{\rd}\), explicit calibration modes for \(\Theta_{\rm SN}\), low-rank growth/lensing amplitudes for \(\Theta_{\rm pert}\), and standard-siren propagation amplitudes for \(\Theta_{\rm GW}\). CMB lensing is
part of the perturbation data block because it reconstructs the projected matter distribution and is routinely used as a growth and lensing probe \cite{LewisChallinor2006,Planck2018Lensing}. GW polarization and strong-field observables can be included as additional consistency checks when the corresponding likelihoods are part of the analysis, but they are not required for
the minimal sector partition. Table~\ref{tab:detectionbasis} summarizes the minimal detection basis and its principal data support.

\begin{table}[!tbp]
\caption{Minimal detection basis used for sector-resolved inference. Here BAO denotes baryon acoustic oscillations, SN supernovae, CMB the cosmic microwave background, FS full-shape clustering, RSD redshift-space distortions, WL weak lensing, and \(E_G\) the large-scale gravity-consistency statistic. Detailed model-family realizations are deferred to Appendix~\ref{app:illustrative_catalog}.}
\label{tab:detectionbasis}
\tablefont
\tighttable
\begin{tabular}{@{}llll@{}}
\toprule
\tcell{0.18\textwidth}{Sector} &
\tcell{0.24\textwidth}{Detection amplitude} &
\tcell{0.26\textwidth}{Principal role} &
\tcell{0.25\textwidth}{Primary data block} \\
\midrule
\tcell{0.18\textwidth}{Late-time background} &
\tcell{0.24\textwidth}{\(\delta\ln E(z)\) modes or \((w_0,w_a)\)} &
\tcell{0.26\textwidth}{Background-shape activation} &
\tcell{0.25\textwidth}{BAO shape, SN, primary-CMB distance anchors} \\
\tcell{0.18\textwidth}{Early-time ruler} &
\tcell{0.24\textwidth}{\(\delta\ln\rd\)} &
\tcell{0.26\textwidth}{Sound-horizon shift} &
\tcell{0.25\textwidth}{BAO scale, primary CMB acoustic scale} \\
\tcell{0.18\textwidth}{SN systematics} &
\tcell{0.24\textwidth}{\(\Delta m_B\), \(q_a\) modes} &
\tcell{0.26\textwidth}{Calibration/selection structure} &
\tcell{0.25\textwidth}{SN Hubble diagram} \\
\tcell{0.18\textwidth}{Perturbations} &
\tcell{0.24\textwidth}{\(\mu-1\), \(\Sigma-1\), \(\eta-1\) low-rank modes} &
\tcell{0.26\textwidth}{Growth/slip activation} &
\tcell{0.25\textwidth}{FS/RSD, WL, CMB lensing, CMB-lensing--galaxy cross-correlations, \(E_G\)} \\
\tcell{0.18\textwidth}{GW propagation} &
\tcell{0.24\textwidth}{\(\ln \Xi(z)\) modes or \((\Xi_0,n)\)} &
\tcell{0.26\textwidth}{Tensor-propagation activation} &
\tcell{0.25\textwidth}{Standard sirens} \\
\bottomrule
\end{tabular}
\end{table}

For sector \(\alpha\), let \(F_\alpha \equiv J_\alpha^T C^{-1}J_\alpha\) be the fiducial Fisher operator of
the compressed block. An admissible detection basis
\(\{\phi_a^\alpha\}_{a=1}^{N_\alpha}\) should satisfy
\begin{equation}
\langle \phi_a^\alpha,\phi_b^\alpha\rangle_{F_\alpha}=\delta_{ab},
\qquad
\langle u,v\rangle_{F_\alpha}\equiv u^T F_\alpha v,
\label{eq:fisher_orthonormality}
\end{equation}
and the retained rank \(N_\alpha\) should obey a fixed capture rule,
\begin{equation}
\frac{\sum_{a=1}^{N_\alpha}\lambda_a^{(\alpha)}}
{\sum_{a\ge 1}\lambda_a^{(\alpha)}} \ge q,
\qquad
q\in(0,1),
\label{eq:capture_rule}
\end{equation}
where \(\lambda_a^{(\alpha)}\) are the ordered eigenvalues of \(F_\alpha\). Basis-robustness tests should compare only bases satisfying Eqs.~(\ref{eq:fisher_orthonormality})--(\ref{eq:capture_rule}). Once the admissible detection basis has been fixed, the active sector amplitudes can be assigned a Fisher-isotropic reference prior.

Because the retained detection modes are orthonormal with respect to the fiducial Fisher metric, the reference prior on the active amplitudes is isotropic in Fisher-normalized coordinates. For sector \(\alpha\), let
\(\vect s_\alpha\in\mathbb R^{N_\alpha}\) collect the retained mode amplitudes.
We then define
\begin{equation}
\pi(\vect s_\alpha\mid A_\alpha=1,\sigma_\alpha)
=
\mathcal N\!\left(\vect 0,\sigma_\alpha^2 I_{N_\alpha}\right),
\qquad
\pi(\vect s_\alpha\mid A_\alpha=0)
=
\delta^{(N_\alpha)}(\vect s_\alpha).
\label{eq:sector_gaussian_prior}
\end{equation}
The hyperparameter \(\sigma_\alpha\) is the active-sector prior width. Prior robustness is assessed by the one-parameter deformation \(\sigma_\alpha\mapsto \rho\,\sigma_\alpha\), which induces the diagnostics \(\eta_\alpha(\rho)\) of Appendix~\ref{app:robustness}.

Under the local quadratic approximation to the compressed likelihood, the sector amplitudes are measured in units of Fisher curvature. In particular,
\begin{equation}
\Delta\chi^2_\alpha
\;\simeq\;
\vect s_\alpha^{\,T} F_\alpha \vect s_\alpha
\;=\;
\vect s_\alpha^{\,T}\vect s_\alpha
\qquad
\text{(Fisher-normalized basis)},
\label{eq:local_chi2_fisher}
\end{equation}
so the isotropic prior of Eq.~(\ref{eq:sector_gaussian_prior}) implies
\begin{equation}
\mathbb E\!\left[\Delta\chi^2_\alpha\mid A_\alpha=1,\sigma_\alpha\right]
=
N_\alpha \sigma_\alpha^2,
\qquad
\mathrm{Var}\!\left[\Delta\chi^2_\alpha\mid A_\alpha=1,\sigma_\alpha\right]
=
2N_\alpha \sigma_\alpha^4.
\label{eq:prior_scale_fisher}
\end{equation}
The width \(\sigma_\alpha\) fixes the prior scale in units of the local curvature of the compressed likelihood. The standardized prior is the coordinate-neutral prior induced by the adopted detection metric.

\subsection{Minimal sector-resolved likelihood specification}
\label{subsec:likelihood}

The BAO and supernova blocks are written in greater detail because they define the dominant current degeneracy between late-time background, ruler, and calibration sectors. The perturbation and tensor sectors are specified only to the level required for a minimal sector-resolved estimator.

Formally, the starting point is the full compressed-data likelihood
\begin{equation}
-2\ln \Ltot
=
(\vect d-\vect \mu_{\rm th})^T C^{-1} (\vect d-\vect \mu_{\rm th})
+\text{const},
\label{eq:full_likelihood}
\end{equation}
where \(\vect d\) collects all compressed observables and \(C\) is their
validated joint covariance. The factorized form
\begin{equation}
\Ltot
=
\Lth_{\rm BAO}\,
\Lth_{\rm SN}\,
\Lth_{\rm CMB}\,
\Lth_{\rm FS/RSD/WL}\,
\Lth_{\rm GW,prop}\,
\Lth_{\rm aux},
\label{eq:Ltot}
\end{equation}
is therefore an approximation whose validity is controlled by the compressed covariance structure. Here \(\Lth_{\rm aux}\) denotes additional consistency likelihoods, including GW polarization and strong-field observables, that are not required for the minimal sector-resolved likelihood specification.

For BAO, the scientifically relevant split should be introduced in the variables used by the observational community. Let the compressed anisotropic BAO vector at redshifts \(z_i\) be
\begin{equation}
\vect d=
\begin{pmatrix}
\vect d_\perp\\
\vect d_\parallel
\end{pmatrix},
\qquad
\vect d_\perp=\left(\ln\frac{\DM(z_i)}{\rd}\right)_{i=1}^N,
\qquad
\vect d_\parallel=\left(\ln\frac{\DHu(z_i)}{\rd}\right)_{i=1}^N,
\label{eq:BAOd}
\end{equation}
where \(\DHu(z)\equiv c/H(z)\). Define the AP-shape and isotropic-scale variables
\begin{equation}
\vect y\equiv \vect d_\perp-\vect d_\parallel,
\qquad
\vect v\equiv \frac{2\vect d_\perp+\vect d_\parallel}{3}.
\label{eq:yv_def_basic}
\end{equation}
Then
\begin{equation}
y_i=\ln\frac{\DM(z_i)}{\DHu(z_i)}=\ln \FAP(z_i),
\qquad
v_i=\ln\left(\frac{\DV(z_i)}{\rd}\right)-\frac{1}{3}\ln z_i,
\label{eq:bao_ap_dv}
\end{equation}
with
\begin{equation}
\DV(z)\equiv [z\DM^2(z)\DHu(z)]^{1/3}.
\label{eq:Dv_definition}
\end{equation}
The additive \(-(1/3)\ln z_i\) in Eq.~(\ref{eq:bao_ap_dv}) is known and drops out of residuals. Thus \(\vect y\) is the standard Alcock--Paczynski variable, while \(\vect v\) is the conventional isotropic BAO scale \(\DV/\rd\) written in logarithmic residual form \cite{AlcockPaczynski1979,Eisenstein2005}.

Let
\begin{equation}
\mathcal T_{yv}
\equiv
\begin{pmatrix}
I_N & -I_N\\
\frac{2}{3}I_N & \frac{1}{3}I_N
\end{pmatrix},
\qquad
\begin{pmatrix}\vect y\\ \vect v\end{pmatrix}
=
\mathcal T_{yv}\vect d,
\qquad
C_{yv}=\mathcal T_{yv} C_d \mathcal T_{yv}^{T},
\label{eq:Tyv_Cyv}
\end{equation}
where \(C_d\) is the covariance of \(\vect d\). Define residuals
\begin{equation}
\Delta\vect y\equiv \vect y-\vect y_{\rm th},
\qquad
\Delta\vect v\equiv \vect v-\vect v_{\rm th}.
\label{eq:bao_residuals}
\end{equation}
The BAO likelihood is the joint Gaussian
\begin{equation}
-2\ln\Lth_{\rm BAO}
=
\begin{pmatrix}
\Delta\vect y\\
\Delta\vect v
\end{pmatrix}^{T}
C_{yv}^{-1}
\begin{pmatrix}
\Delta\vect y\\
\Delta\vect v
\end{pmatrix}.
\label{eq:LBAOjoint}
\end{equation}
The BAO shape/scale split is therefore a decomposition of information content, not an assumption of independence unless the off-diagonal block of \(C_{yv}\) is negligible.

The sensitivity of these variables to a homogeneous sound-horizon rescaling is
\begin{equation}
\frac{\partial \vect y_{\rm th}}{\partial \ln \rd}=\vect 0,
\qquad
\frac{\partial \vect v_{\rm th}}{\partial \ln \rd}=-\one,
\qquad
\frac{\partial \vect v_{\rm th}}{\partial c_a^E}\neq \vect 0
\quad \text{in general}.
\label{eq:bao_sensitivity}
\end{equation}
Thus \(\vect y\) is exactly ruler-free, whereas \(\vect v\) is ruler-sensitive but not ruler-exclusive. BAO alone does not identify \(\Theta_{\rd}\) without external anchors that constrain the late-time background.

The pure-ruler diagnostic must be defined from the null covariance of the scale-like block before it is interpreted. Let \(C_v\) be the lower-right \(N\times N\) block of \(C_{yv}\). The \(C_v^{-1}\)-optimal projection of the scale residual onto the pure-ruler direction is
\begin{equation}
a_{\rd}
\equiv
-\frac{\one^T C_v^{-1}\Delta\vect v}{\one^T C_v^{-1}\one},
\qquad
\Delta\vect v_\perp
\equiv
\Delta\vect v+a_{\rd}\one,
\qquad
\one^T C_v^{-1}\Delta\vect v_\perp=0.
\label{eq:rd_projection}
\end{equation}
Under the null covariance,
\begin{equation}
{\rm Var}(a_{\rd})=(\one^T C_v^{-1}\one)^{-1},
\qquad
\Delta\chi^2_{\rd,{\rm BAO}}=a_{\rd}^{2}\,\one^T C_v^{-1}\one.
\label{eq:bao_ruler_diagnostic}
\end{equation}
A pure-ruler interpretation is supported only when the likelihood gain is localized in \(a_{\rd}\), while \(\vect y\) and \(\Delta\vect v_\perp\) remain statistically stable. If either the AP block or the scale residual orthogonal to the pure-ruler direction moves appreciably, the support is not for \(\Theta_{\rd}\) alone; it necessarily feeds the late-time geometric sector, and curvature if that degree of freedom is open. This convention also makes the internal BAO diagnostic directly comparable to compressed analyses that use \(\FAP\) and \(\DV/\rd\); recent studies indicate that such compressed analyses remain robust across broad classes of non-\lcdm\ and modified-gravity scenarios \cite{Bernal2020BAORobustness,Pan2024CompressedBAO}.

The supernova block must keep the leading calibration directions inside the global inference. A convenient formulation writes the residual vector as
\begin{equation}
\vect r_\mu
\equiv
\vect\mu_{\rm obs}-\vect\mu_{\rm th}(\Theta)
=
M\vect q+\vect\epsilon,
\qquad
\vect\epsilon\sim\mathcal N(\vect 0,C_\mu),
\label{eq:rmu}
\end{equation}
with \(M\) spanning the absolute-magnitude calibration direction and whitened systematic eigenmodes. If the nuisance prior is Gaussian, \(\vect q\sim\mathcal N(\vect 0,\Pi_q)\), then analytic marginalization yields
\begin{equation}
-2\ln\Lth_{\rm SN}
=
\vect r_\mu^{\,T} C_{\rm eff}^{-1} \vect r_\mu
+
\ln\det C_{\rm eff}
+\text{const},
\qquad
C_{\rm eff}=C_\mu+M\Pi_q M^T.
\label{eq:LSNmarg}
\end{equation}
This form is preferable for evidence calculations because the calibration sector remains inside the model comparison rather than being absorbed into preprocessing. That choice is particularly important in the DESI era, where recent reanalyses have shown that the apparent evidence for evolving dark energy can depend materially on supernova calibration structure
\cite{DESDovekie2025,HuangCaiWang2025,OngYallupHandley2026}.

The same point can be made explicitly for the SN block.
Under the Gaussian hierarchical model above, the posterior mean calibration amplitude is
\begin{equation}
\widehat{\vect q}
=
\left(M^{T}C_\mu^{-1}M+\Pi_q^{-1}\right)^{-1}
M^{T}C_\mu^{-1} r_\mu,
\label{eq:sn_posterior_mean}
\end{equation}
with calibration-explained component \(M\widehat{\vect q}\) and orthogonal
residual
\begin{equation}
r_\mu^{\perp}
\equiv
r_\mu - M\widehat{\vect q}.
\label{eq:sn_orthogonal_residual}
\end{equation}
This is what turns the marginalized SN likelihood into a sector diagnostic. If the improvement obtained by activating \(\Theta_{\rm SN}\) is carried predominantly by \(M\widehat{\vect q}\) while \(r_\mu^{\perp}\) remains statistically stable, the support should accrue to the calibration sector rather than to \(\Theta_{E(z)}\). Conversely, a persistent structured shift in \(r_\mu^{\perp}\) after the leading calibration modes have been activated is evidence that the anomaly is not exhausted by low-rank SN systematics.

In practice, the calibration design matrix should be constructed from the survey systematic covariance. Writing
\begin{equation}
C_{\rm sys}=U\Lambda U^T,
\end{equation}
we take
\begin{equation}
M=[\one,\ E_r],
\qquad
E_r \equiv U_r \Lambda_r^{1/2},
\label{eq:sn_mode_construction}
\end{equation}
where \(U_r\) contains the first \(r\) eigenvectors and \(r\) is the smallest rank satisfying
\begin{equation}
\frac{\sum_{i=1}^{r}\lambda_i}{\sum_{i\ge 1}\lambda_i}\ge 0.95.
\label{eq:sn_capture}
\end{equation}
The nuisance prior is then standardized as
\(\vect q\sim\mathcal N(\vect 0,I_{r+1})\), which makes the supernova sector definition reproducible and separates calibration rank choice from cosmological parameterization.

The same logic applies to perturbation and tensor-sector inference. Growth and lensing are represented in a low-rank phenomenological basis \(\{\mu,\Sigma,\eta\}\). The minimal tensor-sector block is GW propagation; GW polarization and strong-field observables may be added as consistency likelihoods when available. The key point is not the unique choice of basis, but the separation of sector identification from microphysical interpretation.

\subsection{Numerical and methodological requirements}
\label{subsec:implementation}

An analysis using the grouped construction separates two computational tasks. The marginal-likelihood calculations dominate the numerical cost; once the catalog evidences have been computed, the quotient aggregation in Eq.~(\ref{eq:aggregation_cost}) has negligible computational cost. The demanding step is the construction and validation of a pattern-labeled catalog whose inactive sectors are inactive at the likelihood level. Numerically, log-evidence uncertainties must be propagated to grouped quantities. The parameter-to-amplitude maps, activation switches, detection bases, likelihood blocks, and emulator interfaces must be defined consistently across all catalog elements.

Because posterior sector odds and grouped Bayes factors are nonlinear functions of the catalog evidences, numerical log-evidence errors must be propagated to grouped quantities. Let
\(\ell_i\equiv \ln Z_i\). Then
\begin{equation}
\ln Z_{\Act}
=
\log\sum_{i\in\Cat_{\Act}}\pi(i\mid \Act)\,e^{\ell_i},
\qquad
\frac{\partial \ln Z_{\Act}}{\partial \ell_i}
=
\rho_{i\mid \Act}
\equiv
\frac{\pi(i\mid \Act)Z_i}{Z_{\Act}}.
\label{eq:grouped_weights}
\end{equation}
If \(\Sigma_{\Act}^{(\ell)}\) is the covariance of the vector of evidence estimates \(\{\hat\ell_i\}_{i\in\Cat_{\Act}}\), the delta method gives
\begin{equation}
{\rm Var}(\widehat{\ln Z_{\Act}})
\simeq
\vect \rho_{\Act}^{\,T}\Sigma_{\Act}^{(\ell)}\vect\rho_{\Act}.
\label{eq:var_lnZAct}
\end{equation}
The same linearization yields a closed-form grouped-Bayes-factor uncertainty. For sector \(\alpha\) define
\begin{align}
u_{\Act}^{+,\alpha}
&=
\frac{\one\!\{A_\alpha(\Act)=1\}\,Z_{\Act}\pi(\Act\mid A_\alpha=1)}
{\sum_{\Act':A_\alpha(\Act')=1}Z_{\Act'}\pi(\Act'\mid A_\alpha=1)},\\
u_{\Act}^{-,\alpha}
&=
\frac{\one\!\{A_\alpha(\Act)=0\}\,Z_{\Act}\pi(\Act\mid A_\alpha=0)}
{\sum_{\Act':A_\alpha(\Act')=0}Z_{\Act'}\pi(\Act'\mid A_\alpha=0)},
\label{eq:sector_lnB_evidence_weights}
\end{align}
and
\begin{equation}
g_{\Act}^{(\alpha)}=
u_{\Act}^{+,\alpha}-
u_{\Act}^{-,\alpha}.
\label{eq:sector_lnB_gradient}
\end{equation}
If \(\Sigma_{\ActSpace}^{(\ell)}\) is the covariance matrix of the pattern-level log-evidence estimates, including any common Monte Carlo or emulator errors, then
\begin{equation}
{\rm Var}(\widehat{\ln B_\alpha})
\simeq
\sum_{\Act,\Act'}
g_{\Act}^{(\alpha)}
\Sigma_{\Act\Act'}^{(\ell)}
g_{\Act'}^{(\alpha)}.
\label{eq:lnB_error_propagation}
\end{equation}
For independent nested-sampling runs the matrix is diagonal, but Eq.~(\ref{eq:lnB_error_propagation}) is the safer reporting formula whenever two catalog elements share emulators, nuisance calibrations, covariance estimates, or importance samples. For an analysis using this construction we recommend
\begin{equation}
\sigma(\ln B_\alpha)<0.3,
\qquad
\max_{r,r'}\left|
\widehat{\ln B_\alpha}^{\,(r)}-
\widehat{\ln B_\alpha}^{\,(r')}
\right|<0.2,
\label{eq:evidence_precision_targets}
\end{equation}
where \(r\) labels at least three independent sampling runs with different random seeds. These targets are deliberately more stringent than the usual ``substantial evidence'' scale: numerical uncertainty in \(\ln B_\alpha\) should be small compared with the evidence differences used to rank competing activation patterns.

The marginal-likelihood calculation is part of the statistical specification. Nested sampling is a natural choice because it returns both posterior samples and estimates of \(Z_i\), but the quotient construction only requires calibrated marginal likelihoods and their numerical uncertainties. For every catalog element an analysis run should report the sampler, stopping criterion, number of live points or equivalent resolution parameter, posterior effective sample size, random seed, \(\widehat{\ln Z_i}\), \(\sigma(\ln Z_i)\), and at least one independent rerun. In multimodal or strongly curved degeneracy cases, representative evidences should be checked with an independent sampling
configuration, for example MultiNest, PolyChord, or dynesty
\cite{Skilling2006Nested,Feroz2009MultiNest,Handley2015PolyChord,Speagle2020Dynesty}. If two validated configurations differ by
\begin{equation}
|\Delta\ln Z_i|>\max\{0.5,2\sigma_{\Delta\ln Z_i}\},
\label{eq:evidence_estimator_disagreement}
\end{equation}
then the affected grouped Bayes factors and posterior sector odds should be treated as numerically unresolved and not used for a physical sector claim until the discrepancy is resolved.

For \(K\) non-baseline sectors the maximal number of activation patterns is \(|\ActSpace|\le 2^K\). The five-sector partition used here has \(K=5\) and hence at most \(32\) patterns. The geometrically dominant DESI subproblem
\((\Theta_{E(z)},\Theta_{\rd},\Theta_{\rm SN})\) has only \(2^3=8\) patterns. The number of catalog evidences is
\begin{equation}
N_Z=\sum_{\Act\in\ActSpace}|\Cat_{\Act}|,
\qquad
N_{\rm like}^{\rm tot}
\simeq
\sum_{\Act\in\ActSpace}\sum_{i\in\Cat_{\Act}}N_{\rm like}^{(i,\Act)},
\label{eq:implementation_cost}
\end{equation}
where \(N_{\rm like}^{(i,\Act)}\) is the number of likelihood calls required by the evidence calculation for catalog element \(i\) in pattern \(\Act\). The aggregation step
\begin{equation}
Z_{\Act}=\sum_{i\in\Cat_{\Act}}Z_i\pi(i\mid\Act)
\label{eq:aggregation_cost}
\end{equation}
is \(O(N_Z)\) and is negligible compared with the evidence evaluations.

A staged analysis is preferable to beginning with the full five-sector catalog.
The recommended sequence is
\begin{enumerate}
\item geometry-only:
\((\Theta_{E(z)},\Theta_{\rd},\Theta_{\rm SN})\), at most eight patterns;
\item geometry plus perturbations:
\((\Theta_{E(z)},\Theta_{\rd},\Theta_{\rm SN},\Theta_{\rm pert})\), at most
sixteen patterns;
\item full five-sector catalog including GW propagation, at most thirty-two
patterns.
\end{enumerate}
The first stage is already sufficient to address the principal DESI--CMB--SN geometric degeneracy. Later stages should be added only after null, single-sector, and mixed-sector benchmark calibration has been passed for the preceding stage.

Existing Boltzmann, emulator, and likelihood codes need not be rewritten as a single trans-dimensional code. The sector partition can be imposed through explicit sector switches and likelihood interfaces within existing codes such as Cobaya, CAMB, CLASS, EFTCAMB, and hi\_class
\cite{TorradoLewis2021Cobaya,Lewis2000CAMB,Blas2011CLASS,Hu2014EFTCAMB,
Zumalacarregui2017HICLASS}. Each catalog element should expose the parameter-to-amplitude map
\begin{equation}
{\cal S}_i(\vartheta_i)
=
\{\vect s_E,\,s_{\rd},\,\vect s_{\rm SN},\,\vect s_{\rm pert},\,
\vect s_{\rm GW}\},
\label{eq:sector_adapter}
\end{equation}
and the activation pattern must be assigned by structural support:
\begin{equation}
A_\alpha(i)=0
\quad\Longleftrightarrow\quad
\vect s_\alpha(\vartheta_i)=\vect 0
\quad\forall\,\vartheta_i\in\Vspace_i.
\label{eq:structural_support_implementation}
\end{equation}
For high-dimensional parameterizations this condition should be enforced by explicit sector switches, not by post-processing. For \(A_{E(z)}=0\), the background expansion module is fixed to the baseline history. For \(A_{\rd}=0\), the early-time sound-horizon module is fixed to its baseline value. For
\(A_{\rm SN}=0\), the calibration design matrix is excluded or assigned zero prior width. For \(A_{\rm pert}=0\), modified-growth and lensing amplitudes are zero in both transfer-function generation and LSS prediction, including CMB
lensing if that block is present. For \(A_{\rm GW}=0\), the GW propagation law is fixed to the GR luminosity-distance relation.

Inactive-sector leakage should be tested in Fisher-normalized units. For a catalog element with \(A_\alpha(i)=0\), define
\begin{equation}
\epsilon_{\alpha i}^{\rm leak}
\equiv
\max_{\vartheta_i\in\Vspace_i}
\left(\vect s_\alpha^{\,T}\vect s_\alpha\right)^{1/2}.
\label{eq:sector_leakage}
\end{equation}
Validation runs should require
\begin{equation}
\epsilon_{\alpha i}^{\rm leak}<10^{-3},
\qquad
\Delta\chi^2_{\alpha,{\rm leak}}<10^{-6},
\qquad
|\Delta\ln Z|_{\rm leak}<0.05
\label{eq:leakage_targets}
\end{equation}
on null mocks when inactive-sector settings are toggled. The first two criteria follow from the Fisher normalization in Eq.~(\ref{eq:local_chi2_fisher}); the last is a validation target for the likelihood calculation. These inactive-sector tests are
stricter than the reported evidence precision in  Eq.~(\ref{eq:evidence_precision_targets}) because they diagnose catalog integrity, not statistical support.

The analysis proceeds in five steps. First, construct a pattern-labeled catalog and specify \(\pi(\Act)\) together with
\(\pi(i\mid \Act)\). Second, choose admissible detection bases and construct explicit sector switches that enforce inactive-sector boundaries. Third, compute the catalog evidences \(Z_i\) with a common likelihood specification and 
marginal-likelihood settings, then aggregate them to
\(Z_{\Act}\). Fourth, report \(p(\Act\mid D)\), \(P_\alpha\),
\(P_{\alpha\beta}\), and \(\ln B_\alpha\). Fifth, assess evidence precision, inactive-sector validation, calibration, posterior-predictive adequacy, prior-width sensitivity, sector-partition sensitivity, and catalog-refinement robustness.

\subsection{Minimal geometry-only DESI--CMB--SN analysis}
\label{subsec:geometry_application}

The most direct analysis is the geometry-only catalog, not the full five-sector catalog. It addresses the current DESI-era geometric degeneracy: whether posterior support currently attributed to late-time dark energy is instead absorbed by a sound-horizon shift or by supernova calibration structure.  The
three-sector pattern space is
\begin{equation}
\Act_g=(A_{E(z)},A_{\rd},A_{\rm SN})\in\{0,1\}^3,
\qquad |\ActSpace_g|=8.
\label{eq:geometry_pattern_space}
\end{equation}
For each supernova likelihood or compilation \(X\), analyze the common data vector
\begin{equation}
D_X=
\{D_{\rm BAO}^{\rm DESI\,DR2},
D_{\rm CMB}^{\rm Planck/ACT},
D_{\rm SN}^{X}\},
\qquad
X\in\{\mathrm{Pantheon+},\mathrm{Union3},\mathrm{DES\!-\!SN5YR},\mathrm{Dovekie}\}.
\label{eq:current_geometry_stack}
\end{equation}
The BAO block should include the DESI DR2 galaxy/quasar and Ly\(\alpha\) BAO measurements, the CMB block may be run with Planck 2018 as the reference and ACT DR6 as a high-precision cross-check, and the SN block should be swapped among Pantheon+, Union3, DES-SN5YR, and recalibrated DES-Dovekie likelihoods when available \cite{DESIDR2BAO2025,DESIDR2Lya2025,Planck2018Parameters,
ACTDR6Extended2025,BroutPantheonPlus2022,RubinUnion32023,DESY5SN2024,
DESDovekie2025}.  The same catalog, pattern prior, Fisher basis, covariance convention, marginal-likelihood estimation settings, and nuisance-prior convention must be used for all
\(X\).  Only the SN likelihood block is replaced.  This controlled comparison is the minimal observational test for distinguishing a stable late-time sector inference from one driven by the adopted SN calibration model.

The required output for each \(D_X\) is
\begin{equation}
\mathcal O_X=
\{p(\Act_g\mid D_X),P_{E(z)},P_{\rd},P_{\rm SN},
P_{E,\rd},P_{E,{\rm SN}},P_{\rd,{\rm SN}},
\ln B_{E(z)},\ln B_{\rd},\ln B_{\rm SN},
\Delta\vect y,a_{\rd},\Delta\vect v_\perp,
\widehat{\vect q},\vect r_\mu^\perp\}.
\label{eq:geometry_application_outputs}
\end{equation}
The diagnostic interpretation is as follows.  A late-time-background interpretation requires stable posterior mass on patterns with \(A_{E(z)}=1\) under \(X\to Y\) SN replacement and a residual structure not removed by
\(a_{\rd}\) or by \(M\widehat{\vect q}\).  A ruler interpretation requires localization in \(a_{\rd}\), stable AP residuals \(\Delta\vect y\), stable orthogonal scale residuals \(\Delta\vect v_\perp\), and stability under SN
substitution.  A calibration interpretation requires posterior mass to move primarily into patterns with \(A_{\rm SN}=1\), with the improvement carried by
\(M\widehat{\vect q}\) and without SN-block PPC failure.

For two SN choices \(X\) and \(Y\), define the difference between pattern posteriors and its total-variation norm by
\begin{equation}
\Delta_{X\to Y}(\Act_g)
\equiv
p(\Act_g\mid D_Y)-p(\Act_g\mid D_X),
\qquad
\TV_{X\to Y}
\equiv
\frac12\sum_{\Act_g\in\ActSpace_g}
|\Delta_{X\to Y}(\Act_g)|.
\label{eq:pattern_posterior_tv}
\end{equation}
For \(\TV_{X\to Y}\lesssim0.01\), the sector posterior is stable under that SN substitution.  For larger shifts, the table of \(\Delta_{X\to Y}(\Act_g)\) should be inspected directly.  A calibration-driven response requires the dominant positive changes to occur in patterns with \(A_{\rm SN}=1\), an increase in \(P_{\rm SN}\) and/or \(\ln B_{\rm SN}\), localization in \(M\widehat{\vect q}\), and acceptable SN PPC.  A late-time response instead requires the dominant positive changes to occur in patterns with \(A_{E(z)}=1\), with residual structure not explained by \(a_{\rd}\) or by the calibration modes.  A ruler response requires dominant positive changes in patterns with \(A_{\rd}=1\) and localization in \(a_{\rd}\).

\begin{table}[!tbp]
\caption{Minimal eight-pattern geometry analysis for the current DESI--CMB--SN comparison. The table defines the minimum reporting map for the geometry-only analysis; mixed pairs and the triple pattern are analyzed by the same diagnostic and failure-flag logic.}
\label{tab:geometry_application}
\tablefont
\tighttable
\begin{tabular}{@{}llll@{}}
\toprule
\tcell{0.16\textwidth}{Open sector(s)} &
\tcell{0.20\textwidth}{Representative catalog elements} &
\tcell{0.28\textwidth}{Required localization diagnostic} &
\tcell{0.24\textwidth}{Failure flag} \\
\midrule
\tcell{0.16\textwidth}{None} &
\tcell{0.20\textwidth}{Baseline \lcdm\ or \nulcdm} &
\tcell{0.28\textwidth}{All diagnostics consistent with null mocks} &
\tcell{0.24\textwidth}{False activation in \(>5\%\) of null mocks at \(P_\alpha>0.9\)} \\
\tcell{0.16\textwidth}{\(E(z)\)} &
\tcell{0.20\textwidth}{\wcdm, \wacdm, low-rank \(\delta\ln E\) modes} &
\tcell{0.28\textwidth}{Structured residual not removed by \(a_{\rd}\) or \(M\widehat{\vect q}\)} &
\tcell{0.24\textwidth}{Preference disappears under SN-likelihood substitution or PPC fails} \\
\tcell{0.16\textwidth}{\(\rd\)} &
\tcell{0.20\textwidth}{Sound-horizon shift, restricted early-time ruler families} &
\tcell{0.28\textwidth}{Gain localized in \(a_{\rd}\), with stable \(\Delta\vect y\) and \(\Delta\vect v_\perp\)} &
\tcell{0.24\textwidth}{AP or ruler-orthogonal scale residual must move} \\
\tcell{0.16\textwidth}{SN} &
\tcell{0.20\textwidth}{Low-rank calibration and selection modes} &
\tcell{0.28\textwidth}{Gain carried by \(M\widehat{\vect q}\), with stable \(\vect r_\mu^\perp\)} &
\tcell{0.24\textwidth}{Residual structure remains after calibration modes or SN PPC fails} \\
\tcell{0.16\textwidth}{Mixed pairs/triple} &
\tcell{0.20\textwidth}{Two- or three-sector combinations of the above} &
\tcell{0.28\textwidth}{Nonzero co-activations \(P_{\alpha\beta}\) and block-localized gains} &
\tcell{0.24\textwidth}{A single-sector label is assigned to mixed-sector injections} \\
\bottomrule
\end{tabular}
\end{table}

The quantities in Eqs.~(\ref{eq:patternposterior})--(\ref{eq:sectorodds}) are determined by the catalog evidences, the pattern prior, the within-pattern catalog priors, and the covariance of the log-evidence estimates entering
Eq.~(\ref{eq:lnB_error_propagation}). An analysis using this construction therefore specifies the activation map \(i\mapsto \Act(i)\), the catalog elements in each \(\mathcal M_{\Act}\), the priors \(\pi(\Act)\) and \(\pi(i\mid\Act)\), the likelihood blocks entering the data vector, and the marginal-likelihood estimates \(\{\ln Z_i\}\) with their numerical covariance. These inputs determine \(Z_{\Act}\), \(p(\Act\mid D)\), \(P_\alpha\),
\(P_{\alpha\beta}\), and \(\ln B_\alpha\), and they fix the statistical specification of the sector-level comparison.

\section{Validation criteria for sector-level claims}
\label{sec:validation}

A sector-level claim requires more than a favorable \(\ln B_\alpha\). Because the grouped estimator is intended to identify the physical sector responsible for an apparent anomaly, its outputs must be validated beyond family-by-family model preference. The required checks are calibration on
controlled injections, false-activation control under null mocks,
posterior-predictive adequacy in the data block that drives the preference, robustness to prior widths and catalog refinement, and explicit sensitivity to the declared sector partition.

To demonstrate that the grouped estimator does more than remove catalog multiplicity, a sector-level analysis should be validated on controlled injections that span the actual DESI degeneracy structure. At minimum one should include
null, \(E(z)\)-only, \(\rd\)-only, SN-only, and mixed \((E(z),\rd)\),
\((E(z),{\rm SN})\), and \((\rd,{\rm SN})\) benchmark classes. For each class one should report the recovered \(p(\Act\mid D)\), the dominant-pattern confusion matrix
\begin{equation}
\mathcal C_{\Act\Act'}
\equiv
\Pr\!\left(\widehat{\Act}_{\rm MAP}=\Act'\,\middle|\,\Act_{\rm true}=\Act\right),
\label{eq:pattern_confusion}
\end{equation}
the sector summaries \(\{P_\alpha,\ln B_\alpha\}\), and block-level posterior-predictive \(p\)-values. The key scientific requirement is localization: signals injected in distinct sectors should not systematically collapse onto the same preferred activation pattern. The corresponding validation classes are summarized in Table~\ref{tab:benchmark_suite}.

\begin{table}[!tbp]
\caption{Validation benchmark classes for a first sector-resolved DESI analysis. The entries specify mock-injection and robustness tests for calibrating sector probabilities, false activations, and pattern localization.}
\label{tab:benchmark_suite}
\tablefont
\tighttable
\begin{tabular}{@{}llll@{}}
\toprule
\tcell{0.12\textwidth}{Class} &
\tcell{0.24\textwidth}{Injection content} &
\tcell{0.25\textwidth}{Failure mode probed} &
\tcell{0.26\textwidth}{Required outputs} \\
\midrule
\tcell{0.12\textwidth}{Null} &
\tcell{0.24\textwidth}{Baseline-only mocks} &
\tcell{0.25\textwidth}{False activations and prior-driven sector odds} &
\tcell{0.26\textwidth}{\(\FAR_\alpha\), \(\ECE_\alpha\), \(p_{\rm PPC}^{(\alpha)}\), \(\Gamma_\alpha\)} \\
\tcell{0.12\textwidth}{Single-sector} &
\tcell{0.24\textwidth}{\(E(z)\)-only, \(\rd\)-only, SN-only, perturbation-only} &
\tcell{0.25\textwidth}{Sector mis-localization} &
\tcell{0.26\textwidth}{\(\mathcal C_{\Act,\Act'}\), \(P_\alpha\), \(\ln B_\alpha\), diagnostic residuals} \\
\tcell{0.12\textwidth}{Mixed-sector} &
\tcell{0.24\textwidth}{\((E(z),\rd)\), \((E(z),{\rm SN})\), \((\rd,{\rm SN})\), \((E(z),{\rm pert})\)} &
\tcell{0.25\textwidth}{Composite/single-sector confusion and geometric degeneracy} &
\tcell{0.26\textwidth}{\(p(\Act\mid D)\), \(P_{\alpha\beta}\), \(p_{\rm PPC}^{(\alpha)}\), tied-evidence flags} \\
\tcell{0.12\textwidth}{Robustness} &
\tcell{0.24\textwidth}{Prior-width, sparsity, basis, sector partition, and refinement scans on the above} &
\tcell{0.25\textwidth}{Specification dependence} &
\tcell{0.26\textwidth}{\(\mathcal R_\alpha\), \(D_{\rm JS}^{\rm part}\), \(\Delta_\alpha^\lambda\), \(\zeta_\alpha^{\rm ref}\)} \\
\bottomrule
\end{tabular}
\end{table}

\subsection{Minimum diagnostic suite}
\label{subsec:diagnostic_suite}

A sector-level analysis should report, at minimum, four classes of diagnostics. First, sector probabilities must be calibrated under null, single-sector, and mixed-sector benchmark classes. Second, false activations must be controlled at high posterior-probability thresholds. Third, a favorable grouped sector Bayes factor must not be accompanied by posterior-predictive failure in the same likelihood block. Fourth, the preferred sector must remain stable under moderate deformations of prior widths, admissible detection bases, sector-preserving catalog refinements, and declared sector partition.

A compact diagnostic core includes false-activation rate (FAR) and expected calibration error (ECE),
\begin{equation}
\FAR_\alpha(\tau)
\equiv
\Pr\!\left(P_\alpha>\tau\,\middle|\,\text{null benchmark}\right),
\label{eq:far_main}
\end{equation}
\begin{equation}
\ECE_\alpha
=
\sum_{k=1}^{K}
\frac{n_k}{N}
\left|
\frac{1}{n_k}\sum_{j\in B_k}P_{\alpha,j}
-
\frac{1}{n_k}\sum_{j\in B_k}y_{\alpha,j}
\right|,
\label{eq:ece_main}
\end{equation}
together with a sector-level posterior-predictive-check (PPC) tail probability \(p_{\rm PPC}^{(\alpha)}\) in the likelihood blocks that dominate the inferred sector preference \cite{GelmanMengStern1996PPC}. Here \(y_{\alpha,j}\in\{0,1\}\) denotes whether sector \(\alpha\) is truly active in benchmark realization \(j\). These quantities distinguish calibrated sector-level posterior probabilities from spurious high-posterior activations.

For finite validation suites the uncertainty on \(\FAR_\alpha\) should be reported with a binomial confidence interval, not as a point estimate alone.  A minimum numerical setup for a first observational analysis is
\begin{equation}
N_{\rm null}\ge 500,
\qquad
N_{\rm inj}(\Act)\ge100
\quad\text{for each non-null injected pattern used in the claim},
\label{eq:mock_count_recommendation}
\end{equation}
because \(N_{\rm null}=500\) gives a binomial standard error
\(\sqrt{0.05\times0.95/500}=0.0097\) at the target false-activation rate \(0.05\).  If no false activations are observed, the corresponding one-sided 95\% upper limit should be quoted rather than reported as zero.

Calibration should also be summarized by proper scoring rules \cite{GneitingRaftery2007},
\begin{equation}
\BS_\alpha
=
\frac1N\sum_{j=1}^{N}(P_{\alpha,j}-y_{\alpha,j})^2,
\label{eq:brier_score}
\end{equation}
\begin{equation}
\NLS_\alpha
=-\frac1N\sum_{j=1}^{N}
\left[
y_{\alpha,j}\ln(P_{\alpha,j}+\epsilon_{\rm mach})+
(1-y_{\alpha,j})\ln(1-P_{\alpha,j}+\epsilon_{\rm mach})
\right],
\label{eq:negative_log_score}
\end{equation}
where \(\BS\) penalizes overconfident mis-localization and \(\NLS\) is sensitive to high-confidence false sector claims.

For a block \(b\) driving a sector preference, the posterior-predictive tail probability should be computed from an explicit discrepancy statistic, for example
\begin{equation}
T_b(d_b,\theta)
=
[d_b-\mu_b(\theta)]^T C_b^{-1}[d_b-\mu_b(\theta)],
\label{eq:block_ppc_statistic}
\end{equation}
\begin{equation}
p_{{\rm PPC},b}
=
\Pr\left[T_b(d_b^{\rm rep},\theta)\ge T_b(d_b^{\rm obs},\theta)
\middle|D\right].
\label{eq:block_ppc_tail}
\end{equation}
The sector PPC \(p_{\rm PPC}^{(\alpha)}\) used below is the PPC of the block, or minimum over blocks, that carries the largest posterior-mean contribution to the sector Bayes-factor gain.

The diagnostic suite is part of the statistical specification of the sector claim. Results that fail null control, probability calibration, posterior-predictive adequacy, leakage control, or robustness checks should be reported as unlocalized model-preference signals rather than as sector identifications.

The validation budget should be specified quantitatively. If zero false activations are observed in \(N_{\rm null}\) null realizations, the one-sided 95\% Clopper--Pearson upper bound is \cite{ClopperPearson1934}
\begin{equation}
p_{95}^{(0)}=1-0.05^{1/N_{\rm null}}.
\label{eq:zero_false_upper_bound}
\end{equation}
Thus \(N_{\rm null}\ge59\) is the absolute minimum needed to demonstrate \(\FAR<0.05\) when no false activations occur; in practice \(N_{\rm null}\ge200\) is a useful target, giving binomial standard error \(\sqrt{p(1-p)/N}\simeq0.015\) at \(p=0.05\). For each single-sector and mixed-sector benchmark class, \(N\ge100\) realizations give \(\lesssim0.05\) binomial uncertainty for confusion-matrix entries away from the boundaries.

\subsection{Data-block removal and replacement diagnostics}
\label{subsec:block_diagnostics}

Sector activation should be localized in both parameter space and the data blocks that generate the evidence. Let \(D_{\setminus b}\) denote the data vector with block \(b\) removed, and let \(D[b\to b']\) denote a controlled replacement of one block by an alternative reduction or calibration analysis. We define the leave-one-block-out response
\begin{equation}
\Delta_{\alpha}^{(-b)}\equiv\ln B_\alpha(D)-\ln B_\alpha(D_{\setminus b})
\label{eq:block_removal}
\end{equation}
and the replacement response
\begin{equation}
R_{\alpha}^{(b\to b')}\equiv\ln B_\alpha(D[b\to b'])-\ln B_\alpha(D).
\label{eq:block_replacement}
\end{equation}
The propagated uncertainty is obtained from the same evidence covariance used in Eq.~(\ref{eq:lnB_error_propagation}). For example,
\begin{equation}
\sigma^2(\Delta_{\alpha}^{(-b)})
\simeq
\sigma^2[\ln B_\alpha(D)]+\sigma^2[\ln B_\alpha(D_{\setminus b})]
-2\,\mathrm{Cov}[\ln B_\alpha(D),\ln B_\alpha(D_{\setminus b})],
\label{eq:block_response_variance}
\end{equation}
with an analogous expression for \(R_{\alpha}^{(b\to b')}\). The sector Bayes factor is block-sensitive if
\begin{equation}
|\Delta_{\alpha}^{(-b)}|>\max\{1,2\sigma(\Delta_{\alpha}^{(-b)})\}
\quad\text{or}\quad
{\rm sign}\,\ln B_\alpha(D)\ne {\rm sign}\,\ln B_\alpha(D_{\setminus b}),
\label{eq:block_sensitivity_condition}
\end{equation}
and sensitive to the alternative reduction if the same criterion is triggered by \(R_{\alpha}^{(b\to b')}\). For the DESI geometric problem, the decisive replacement tests are Pantheon+\(\leftrightarrow\)Union3,
Pantheon+\(\leftrightarrow\)DES-SN5YR, and DES-SN5YR\(\leftrightarrow\) Dovekie-like recalibrations at fixed DESI BAO and CMB inputs. A late-time sector claim requires stable \(A_{E(z)}=1\) support under these replacements. A
calibration interpretation requires the dominant response to occur in \(P_{\rm SN}\), \(\ln B_{\rm SN}\), and the calibration diagnostic \(M\widehat{\vect q}\), not in \(P_{E(z)}\) alone.

\subsection{Robustness summary}
\label{subsec:robustness_summary}

A convenient summary of the robustness requirement is
\begin{equation}
\mathcal R_\alpha
\equiv
\left\{
\eta_\alpha(2),\eta_\alpha(4),
\delta_\alpha^{\rm ref},\zeta_\alpha^{\rm ref},
D_{\rm JS}^{\rm basis},D_{\rm JS}^{\rm part},
\Delta_\alpha^\lambda,\Gamma_\alpha,
\max_b |\Delta_\alpha^{(-b)}|,\max_{b\to b'}|R_\alpha^{(b\to b')}|
\right\},
\label{eq:robustness_vector}
\end{equation}
whose explicit definitions are collected in Appendix~\ref{app:technical} and in Sec.~\ref{subsec:sector_partition_prior}. The role of \(\mathcal R_\alpha\) is to distinguish a genuinely supported sector from one that is activated only because of a particular prior width, basis choice, sector-partition split, or catalog discretization.

For an analysis using this construction we recommend the provisional acceptance targets
\begin{equation}
\FAR_\alpha(0.9)<0.05,
\qquad
\ECE_\alpha<0.03,
\qquad
p_{\rm PPC}^{(\alpha)}\in[0.05,0.95],
\qquad
\zeta_\alpha^{\rm ref}<1,
\label{eq:acceptance_targets}
\end{equation}
with additional sector-partition targets given in Eq.~(\ref{eq:sector_partition_targets}). These values are not universal constants, but they define explicit reporting thresholds and place sector-level claims on a quantitative footing.

The numerical thresholds in Eq.~(\ref{eq:acceptance_targets}) are reporting criteria rather than universal decision thresholds. A sector interpretation should quote \(P_\alpha\), \(\ln B_\alpha\), \(\sigma(\ln B_\alpha)\),
\(\FAR_\alpha\), \(\ECE_\alpha\), PPC values, \(\mathcal R_\alpha\), and the prior/sector-partition diagnostics of Sec.~\ref{subsec:sector_partition_prior}. Large
\(P_\alpha\) or favorable \(\ln B_\alpha\) should be interpreted as sector localized only when these calibration, predictive, and robustness checks are satisfied. 

\subsection{Pattern-prior and sector-partition sensitivity}
\label{subsec:sector_partition_prior}

The grouped construction removes arbitrary multiplicity inside a fixed activation pattern while leaving the analyst's declared sector partition explicit. This dependence is statistical: when two or more activation patterns give nearly equal likelihood improvement, the pattern prior is the tie-breaker. Therefore every sector-level analysis must report both a prior scan and a sector-partition coarsening test.

For the sparsity family of Eq.~(\ref{eq:patternprior_main}), the minimal prior scan is
\begin{equation}
\lambda\in\{0,\ln2,2\ln2\},
\label{eq:lambda_scan}
\end{equation}
corresponding to a pattern-neutral prior, a factor-of-two penalty per activated sector, and a stronger factor-of-four penalty per activated sector. We define
\begin{equation}
\Delta_\alpha^\lambda
\equiv
\max_{\lambda,\lambda'\in\{0,\ln2,2\ln2\}}
\left|\ln B_\alpha(\lambda)-\ln B_\alpha(\lambda')\right|,
\qquad
\Delta P_\alpha^\lambda
\equiv
\max_{\lambda,\lambda'}|P_\alpha(\lambda)-P_\alpha(\lambda')|.
\label{eq:lambda_sensitivity}
\end{equation}
Sector-level conclusions should not be regarded as stable if the dominant pattern or leading sector probabilities change qualitatively across this scan.

The dependence on sector boundaries is handled by coarsening-compatible priors. Let \(\widetilde{\ActSpace}\) be a finer sector partition and let
\begin{equation}
C:\widetilde{\ActSpace}\rightarrow \ActSpace
\label{eq:coarsening_map}
\end{equation}
map fine activation patterns to coarse patterns. For example, a split \(\widetilde\Theta_{\rm pert}=\Theta_{\rm growth}\cup\Theta_{\rm lensing}\) is coarsened by
\begin{equation}
A_{\rm pert}=A_{\rm growth}\vee A_{\rm lensing}.
\label{eq:growth_lensing_coarsening}
\end{equation}
The fine prior \(\widetilde\pi\) is coarsening-compatible with the coarse prior \(\pi\) if
\begin{equation}
\sum_{\widetilde\Act:\,C(\widetilde\Act)=\Act}
\widetilde\pi(\widetilde\Act)=\pi(\Act).
\label{eq:coarsening_compatible_prior}
\end{equation}
Equivalently,
\begin{equation}
\widetilde\pi(\widetilde\Act)
=
\pi\big(C(\widetilde\Act)\big)
\rho\big(\widetilde\Act\mid C(\widetilde\Act)\big),
\qquad
\sum_{\widetilde\Act:\,C(\widetilde\Act)=\Act}
\rho(\widetilde\Act\mid\Act)=1.
\label{eq:fine_prior_factorization}
\end{equation}
With this rule, splitting perturbations into growth and lensing does not automatically impose a double high-level penalty on a theory that modifies both. The total prior mass assigned to ``some perturbation activation'' is preserved; only the conditional allocation within that coarse event is changed.

The posterior-level sector-partition diagnostic is the Jensen--Shannon distance between the coarse posterior and the push-forward of the fine posterior,
\begin{equation}
(C_\#\widetilde p)(\Act\mid D)
\equiv
\sum_{\widetilde\Act:\,C(\widetilde\Act)=\Act}
\widetilde p(\widetilde\Act\mid D),
\qquad
D_{\rm JS}^{\rm part}
\equiv
\JS\!\left[p(\Act\mid D),C_\#\widetilde p(\Act\mid D)\right].
\label{eq:sector_partition_js}
\end{equation}
An analysis should report
\begin{equation}
D_{\rm JS}^{\rm part}<0.02\ {\rm nats},
\qquad
\max_\lambda |\Delta P_\alpha^{\rm part}|<0.1,
\qquad
\max_\lambda |\Delta\ln B_\alpha^{\rm part}|<0.5,
\label{eq:sector_partition_targets}
\end{equation}
as provisional targets for sector-partition stability. Here \(\Delta P_\alpha^{\rm part}\) and \(\Delta\ln B_\alpha^{\rm part}\) denote the changes induced by replacing the coarse sector partition with a coarsening-compatible fine sector partition and then pushing the result back to the coarse space.

The prior tie-breaker can also be measured analytically. For \(\pi_\lambda(\Act)\propto \pi_0(\Act)e^{-\lambda K(\Act)}\), define sector odds \(O_\alpha(\lambda)=P_\alpha(\lambda)/[1-P_\alpha(\lambda)]\). Then
\begin{equation}
\frac{\partial}{\partial\lambda}\ln O_\alpha(\lambda)
=
-\mathbb E_{D,\lambda}\!\left[K\mid A_\alpha=1\right]
+\mathbb E_{D,\lambda}\!\left[K\mid A_\alpha=0\right],
\label{eq:prior_tiebreaker_derivative}
\end{equation}
and we define
\begin{equation}
\Gamma_\alpha(\lambda)
\equiv
\left|\frac{\partial}{\partial\lambda}\ln O_\alpha(\lambda)\right|.
\label{eq:Gamma_alpha}
\end{equation}
If \(\Gamma_\alpha\sim 1\), changing \(\lambda\) by one natural unit changes the sector odds by order \(e\). For the geometrically degenerate \((E(z),\rd,{\rm SN})\) subproblem, \(\Gamma_\alpha\) should be reported for all three sectors.

Finally, if two leading activation patterns satisfy
\begin{equation}
\left|\ln Z_{\Act_1}-\ln Z_{\Act_2}\right|
<
\max\left(1,2\sigma_{\Delta\ln Z}\right),
\label{eq:evidence_neartie_condition}
\end{equation}
the leading patterns are not separated at evidence level. In that case a posterior ranking under a particular \(\pi(\Act)\) should not be interpreted as a physical sector detection unless the ranking is stable under Eqs.~(\ref{eq:lambda_sensitivity})--(\ref{eq:sector_partition_targets}) and is localized by the likelihood diagnostics of Sec.~\ref{subsec:reporting_standard}.

\subsection{Reporting standard}
\label{subsec:reporting_standard}

A DESI-era sector-level claim should be reported through the tuple of Eq.~(\ref{eq:reporttuple}). Without that tuple, the result remains a model-preference statement rather than a sector-resolved inference claim. For perturbations, the motivation is aligned with phenomenological large-scale structure descriptions of modified gravity that compress the observable freedom into a small number of functions or amplitudes \cite{BelliniSawicki2014,Ishak2025MGDESI}. For tensor propagation, the \((\Xi_0,n)\) language follows the modified-propagation formalism for standard sirens and its recent observational implementations \cite{Belgacem2018,LVK2026GWTC5}.

Accordingly, a DESI-era analysis should report the global pattern posterior
\(p(\Act\mid D)\) together with, for each relevant sector \(\alpha\),
\begin{equation}
\begin{aligned}
\Big\{&
\text{posterior constraints},\; \ln B_\alpha,\; P_\alpha,\;
\{P_{\alpha\beta}\}_{\beta\neq\alpha},\; \ln S,\; d,\\
&\text{shift diagnostics},\; \text{PPC},\; \text{bias budget},\;
\mathcal R_\alpha,\; \{\Delta_{X\to Y},\TV_{X\to Y}\}
\Big\}.
\end{aligned}
\label{eq:reporttuple}
\end{equation}
The point of this reporting standard is that a DESI-era anomaly should be localized at sector level before it is interpreted in terms of any specific microphysical model class.

For the geometrically dominant three-sector subproblem \((\Theta_{E(z)},\Theta_{\rd},\Theta_{\rm SN})\), the generic tuple above should be supplemented by the internal diagnostic variables
\begin{equation}
\left\{
\Delta\vect y,\;
 a_{\rd},\;
\Delta\vect v_\perp,\;
\widehat{\vect q},\;
\vect r_\mu^{\perp}
\right\}.
\label{eq:diagnostic_variables}
\end{equation}
Here \(\Delta\vect y\), \(a_{\rd}\), and \(\Delta\vect v_\perp\) are the ruler-free, ruler-aligned, and ruler-orthogonal BAO quantities of Eqs.~(\ref{eq:bao_residuals})--(\ref{eq:rd_projection}), while \(\widehat{\vect q}\) and \(\vect r_\mu^{\perp}\) are the posterior calibration mode and orthogonal SN residual of Eqs.~(\ref{eq:sn_posterior_mean})--(\ref{eq:sn_orthogonal_residual}). Ruler support requires localization in \(a_{\rd}\) with stable \(\Delta\vect y\) and \(\Delta\vect v_\perp\). Calibration support requires localization in \(M\widehat{\vect q}\) with stable \(\vect r_\mu^\perp\). Late-time background support requires residual structure not exhausted by either the ruler projection or the calibration modes. These diagnostics determine whether the support is diagnostically localized rather than only preferred in the global evidence.

Result tables should retain the numerical diagnostics rather than replacing them by categorical classifications. In particular, nearly tied pattern evidences, strong \(\lambda\)-dependence, sector-partition dependence, or a large leave-one-data-block-out response should be visible in the table entries themselves.
This convention keeps the statistical assumptions explicit and avoids interpreting a prior tie-breaker as a physical sector detection.

\section{Exact analytic demonstrations}
\label{sec:toychecks}

This section gives two exact demonstrations of the grouped construction in regimes where every relevant quantity is available in closed form. The first isolates the catalog-multiplicity pathology and shows its exact removal by quotient-space inference. The second provides a solvable Gaussian toy catalog with exact pattern posteriors, sector probabilities, grouped Bayes factors, and maximum-a-posteriori (MAP) phase boundaries.

\subsection{Exact duplication and multiplicity bias}
\label{subsec:duplication_bias}

Consider two activation patterns \(\Act_A\) and \(\Act_B\) with equal grouped evidences \(Z_{\Act_A}=Z_{\Act_B}=Z\) and equal pattern priors \(\pi(\Act_A)=\pi(\Act_B)=1/2\). If \(\Act_B\) is represented by \(J\) exact duplicates and one adopts a naive global-uniform prior over raw model labels, then
\begin{equation}
p_{\rm naive}(\Act_B\mid D)=\frac{J}{1+J},
\qquad
p_{\rm naive}(\Act_A\mid D)=\frac{1}{1+J}.
\label{eq:duplication_naive_section}
\end{equation}
By contrast, the grouped construction leaves
\begin{equation}
p(\Act_B\mid D)=p(\Act_A\mid D)=\frac12
\qquad
\text{for all } J\ge 1.
\label{eq:duplication_grouped_section}
\end{equation}
Figure~\ref{fig:duplication_pathology} visualizes this exact pathology and its removal by grouped inference. The example isolates the sector-level aggregation problem: raw model labels bias a later sector statement when they are counted as physical alternatives.

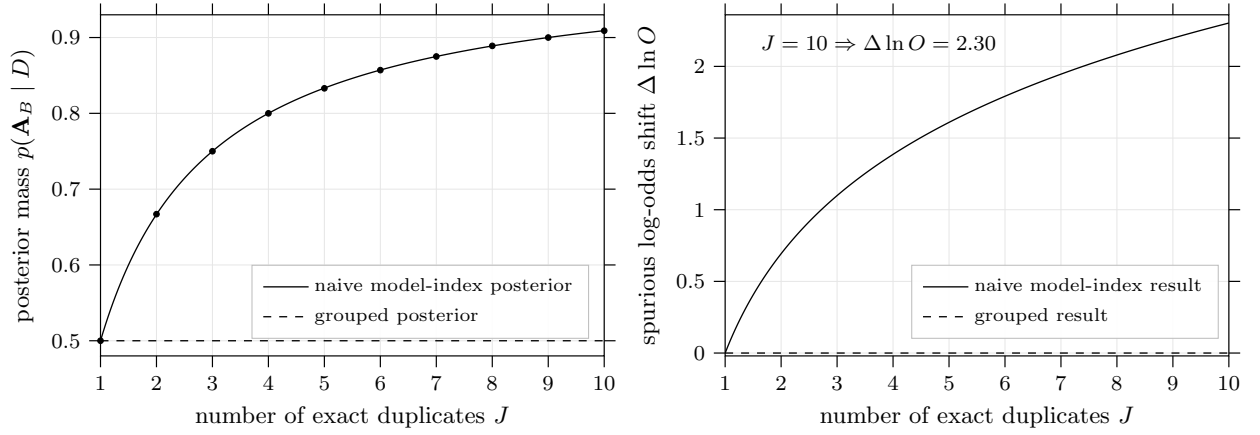
\begin{figure*}[tbp]
\centering
\begin{tikzpicture}
\begin{groupplot}[
  group style={group size=2 by 1, horizontal sep=1.6cm},
  desiAxis,
  height=0.34\textwidth,
  width=0.46\textwidth
]

\nextgroupplot[
  xmin=1, xmax=10,
  ymin=0.48, ymax=0.93,
  xtick={1,2,...,10},
  ytick={0.5,0.6,0.7,0.8,0.9},
  xlabel={number of exact duplicates \(J\)},
  ylabel={posterior mass \(p(\Act_B\mid D)\)},
  legend style={
    at={(0.97,0.05)},
    anchor=south east,
    font=\scriptsize,
    draw=black!25,
    fill=white,
    fill opacity=0.94,
    text opacity=1,
    row sep=1pt,
    inner xsep=4pt,
    inner ysep=3pt
  }
]
\addplot[theorycurve, domain=1:10] {x/(1+x)};
\addlegendentry{naive model-index posterior}
\addplot[datapts, forget plot] coordinates {
  (1,0.500) (2,0.667) (3,0.750) (4,0.800) (5,0.833)
  (6,0.857) (7,0.875) (8,0.889) (9,0.900) (10,0.909)
};
\addplot[refcurve, domain=1:10] {0.5};
\addlegendentry{grouped posterior}

\nextgroupplot[
  xmin=1, xmax=10,
  ymin=-0.02, ymax=2.36,
  xtick={1,2,...,10},
  ytick={0,0.5,1.0,1.5,2.0},
  xlabel={number of exact duplicates \(J\)},
  ylabel={spurious log-odds shift \(\Delta\ln O\)},
  legend style={
    at={(0.98,0.05)},
    anchor=south east,
    font=\scriptsize,
    draw=black!25,
    fill=white,
    fill opacity=0.94,
    text opacity=1,
    row sep=1pt,
    inner xsep=4pt,
    inner ysep=3pt
  }
]
\addplot[theorycurve, domain=1:10] {ln(x)};
\addplot[refcurve, domain=1:10] {0};

\addlegendimage{black, line width=0.55pt}
\addlegendentry{naive model-index result}
\addlegendimage{black, dashed, line width=0.50pt}
\addlegendentry{grouped result}

\node[anchor=south east,font=\footnotesize] at (axis cs:6,2.053)
  {\(J=10 \Rightarrow \Delta\ln O = 2.30\)};

\end{groupplot}
\end{tikzpicture}
\caption{Exact duplication pathology. {Left:} naive model-index inference assigns \(p_{\rm naive}(\Act_B\mid D)=J/(1+J)\) to the activation pattern represented by \(J\) duplicates, whereas the grouped posterior remains exactly \(1/2\). {Right:} the same duplication induces a purely combinatorial log-odds bias \(\Delta\ln O(J)=\logit p_{\rm naive}(\Act_B\mid D)-\logit(1/2)=\ln J\), despite no change in information content.}
\label{fig:duplication_pathology}
\end{figure*}

The induced combinatorial bias is especially transparent in log-odds form:
\begin{equation}
\Delta \ln O(J)
\equiv
\logit p_{\rm naive}(\Act_B\mid D)-\logit p(\Act_B\mid D)
=
\ln J.
\label{eq:duplication_logodds}
\end{equation}
Thus exact duplication injects an additive spurious support term of size \(\ln J\) in favor of \(\Act_B\) without changing the underlying information in the data.

\subsection{Closed-form Gaussian toy catalog}
\label{subsec:gaussian_toy}

A minimal analytic example can be constructed with a whitened compressed observable \(\vect x\in\mathbb R^2\),
\begin{equation}
\vect x = T_{\Act}\vect s_{\Act}+\vect\epsilon,
\qquad
\vect\epsilon\sim\mathcal N(\vect 0,I_2),
\qquad
\vect s_{\Act}\sim \mathcal N(\vect 0,\sigma^2 I_{K_{\Act}}),
\label{eq:toy_model}
\end{equation}
where \(\Act=(A_E,A_{\rd})\in\{00,10,01,11\}\), \(K_{\Act}=A_E+A_{\rd}\), and
\begin{equation}
T_{00}=\emptyset,
\qquad
T_{10}=\vect e_1,
\qquad
T_{01}=\vect e_2,
\qquad
T_{11}=[\vect e_1,\vect e_2].
\label{eq:toy_design}
\end{equation}
The pattern evidence is available in closed form:
\begin{equation}
Z_{\Act}(\vect x)=
\mathcal N\!\left(\vect x;\vect 0,I_2+\sigma^2T_{\Act}T_{\Act}^{T}\right).
\label{eq:toy_evidence}
\end{equation}
We adopt the sparsity prior
\begin{equation}
\pi(\Act)\propto 2^{-K_{\Act}},
\label{eq:toy_pattern_prior}
\end{equation}
and set \(\sigma=2\). Table~\ref{tab:toy_catalog} reports the exact pattern posterior, sector inclusion probabilities, and grouped Bayes factors for representative synthetic observations. The value of this example is exact control: all evidences, pattern posteriors, sector probabilities, grouped Bayes factors, and phase boundaries are analytic. It is therefore the smallest nontrivial catalog in which null, single-sector, and mixed-sector activation can all be displayed exactly.

\begin{table}[!tbp]
\caption{Closed-form Gaussian toy catalog defined by Eqs.~(\ref{eq:toy_model})--(\ref{eq:toy_pattern_prior}). Here \(\Act=(A_E,A_{\rd})\), \(\sigma=2\), and the entries are exact grouped posteriors for representative synthetic observations \(\vect x\).}
\label{tab:toy_catalog}
\tablefont
\tighttable
\begin{tabular}{l c cccc cc cc}
\toprule
Case & \(\vect x\) &
\(p(00\mid \vect x)\) &
\(p(10\mid \vect x)\) &
\(p(01\mid \vect x)\) &
\(p(11\mid \vect x)\) &
\(P_E\) & \(P_{\rd}\) &
\(\ln B_E\) & \(\ln B_{\rd}\) \\
\midrule
Null & \((0,0)\) &
0.668 & 0.149 & 0.149 & 0.033 &
0.183 & 0.183 &
-0.805 & -0.805 \\
\(E\)-only, moderate & \((2,0)\) &
0.388 & 0.429 & 0.087 & 0.096 &
0.526 & 0.183 &
0.795 & -0.805 \\
\(E\)-only, strong & \((3,0)\) &
0.089 & 0.728 & 0.020 & 0.163 &
0.891 & 0.183 &
2.795 & -0.805 \\
\(\rd\)-only, moderate & \((0,2)\) &
0.388 & 0.087 & 0.429 & 0.096 &
0.183 & 0.526 &
-0.805 & 0.795 \\
Mixed & \((2,2)\) &
0.225 & 0.249 & 0.249 & 0.276 &
0.526 & 0.526 &
0.795 & 0.795 \\
\bottomrule
\end{tabular}
\end{table}

\begin{figure}[tbp]
\centering
\begin{tikzpicture}
\def\xstar{1.93511}
\begin{axis}[
  desiAxis,
  grid=none,
  width=0.42\columnwidth,
  height=0.42\columnwidth,
  xmin=-3.4, xmax=3.4,
  ymin=-3.4, ymax=3.4,
  axis equal image,
  xtick={-3,-2,-1,0,1,2,3},
  ytick={-3,-2,-1,0,1,2,3},
  xlabel={$x_E$},
  ylabel={$x_{\rd}$}
]
  \path[fill=black!4]  (axis cs:-\xstar,-\xstar) rectangle (axis cs:\xstar,\xstar);
  \path[fill=black!10] (axis cs:\xstar,-\xstar)  rectangle (axis cs:3.4,\xstar);
  \path[fill=black!10] (axis cs:-3.4,-\xstar)    rectangle (axis cs:-\xstar,\xstar);
  \path[fill=black!10] (axis cs:-\xstar,\xstar)  rectangle (axis cs:\xstar,3.4);
  \path[fill=black!10] (axis cs:-\xstar,-3.4)    rectangle (axis cs:\xstar,-\xstar);

  \path[fill=black!18] (axis cs:\xstar,\xstar)   rectangle (axis cs:3.4,3.4);
  \path[fill=black!18] (axis cs:-3.4,\xstar)     rectangle (axis cs:-\xstar,3.4);
  \path[fill=black!18] (axis cs:-3.4,-3.4)       rectangle (axis cs:-\xstar,-\xstar);
  \path[fill=black!18] (axis cs:\xstar,-3.4)     rectangle (axis cs:3.4,-\xstar);

  \addplot[black, line width=0.45pt] coordinates {(-\xstar,-3.4) (-\xstar,3.4)};
  \addplot[black, line width=0.45pt] coordinates {(\xstar,-3.4) (\xstar,3.4)};
  \addplot[black, line width=0.45pt] coordinates {(-3.4,-\xstar) (3.4,-\xstar)};
  \addplot[black, line width=0.45pt] coordinates {(-3.4,\xstar) (3.4,\xstar)};

  \node[font=\footnotesize] at (axis cs:0,0) {$00$};
  \node[font=\footnotesize] at (axis cs:2.65,0) {$10$};
  \node[font=\footnotesize] at (axis cs:-2.65,0) {$10$};
  \node[font=\footnotesize] at (axis cs:0,2.65) {$01$};
  \node[font=\footnotesize] at (axis cs:0,-2.65) {$01$};
  \node[font=\footnotesize] at (axis cs:2.65,2.65) {$11$};
  \node[font=\footnotesize] at (axis cs:-2.65,2.65) {$11$};
  \node[font=\footnotesize] at (axis cs:-2.65,-2.65) {$11$};
  \node[font=\footnotesize] at (axis cs:2.65,-2.65) {$11$};
\end{axis}
\end{tikzpicture}
\caption{Exact MAP phase diagram of the closed-form Gaussian toy catalog for \(\sigma=2\) and \(\pi(\Act)\propto 2^{-K_{\Act}}\). The decision boundary is \(|x_\alpha|=x_\star\) with \(x_\star^2=\frac{5}{2}\ln(2\sqrt5)=3.7447\), hence \(x_\star=1.935\). The central square is the null pattern \(00\), the side bands are the single-sector patterns \(10\) and \(01\), and the four corners correspond to mixed activation \(11\).}
\label{fig:toy_phase_diagram}
\end{figure}
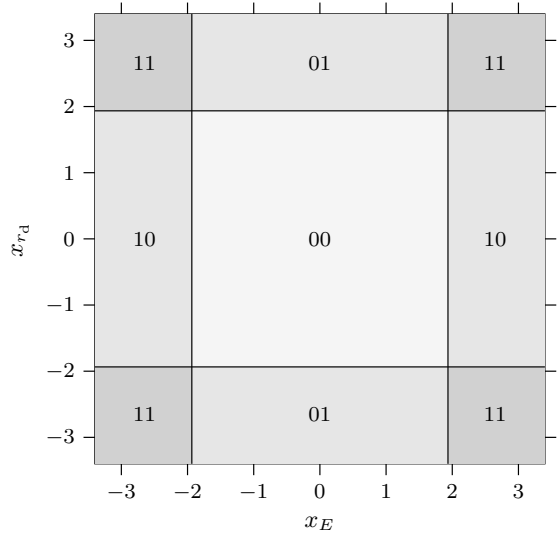

For the adopted prior \(\pi(\Act)\propto 2^{-K_{\Act}}\) and \(\sigma=2\), the MAP phase boundaries are analytic. Comparing the \(10\) and \(00\) patterns gives
\begin{equation}
|x_E|=x_\star,
\qquad
x_\star^2=\frac{5}{2}\ln(2\sqrt5)=3.7447,
\qquad
x_\star=1.935,
\label{eq:toy_threshold}
\end{equation}
and, by symmetry, the same threshold applies to \(|x_{\rd}|\). Hence the null pattern dominates inside the central square \(|x_E|<x_\star\), \(|x_{\rd}|<x_\star\), the single-sector patterns dominate in the side bands, and the mixed pattern dominates in the four corners. This global structure is shown in Fig.~\ref{fig:toy_phase_diagram}.

The two-sector toy is deliberately minimal. It isolates the algebra of grouped activation inference in the smallest catalog that already exhibits null, single-sector, and mixed-sector structure. Extending the same construction to \(\Theta_{\rm SN}\), \(\Theta_{\rm pert}\), and \(\Theta_{\rm GW}\) is straightforward but would add notation rather than new inference logic.

\section{Discussion and conclusions}
\label{sec:concl}

The central statistical object in the construction is the activation pattern \(\Act\), not the raw catalog label. Ordinary pairwise comparisons between predeclared families remain summarized by the corresponding evidence ratio. The quotient construction addresses the sector-level question: how posterior support should be assigned when the scientific statement is a coarse physical event, such as activation of late-time background physics or SN calibration structure, and the catalog contains unequal numbers of parameterizations for different events.

The primary Bayesian object is the activation pattern \(\Act\), and the physically interpretable outputs are derived from its posterior: the sector marginals \(P_\alpha\), the co-activation probabilities \(P_{\alpha\beta}\), and the grouped support quantities \(\ln B_\alpha(D)\). Once sector-preserving catalog refinement leaves the grouped evidences \(Z_{\Act}\) and pattern prior \(\pi(\Act)\) fixed, the sector-level inference is provably unchanged. By contrast, the exact duplication example of Sec.~\ref{sec:toychecks} shows that naive model-index inference suffers an additive combinatorial log-odds bias \(\Delta\ln O(J)=\ln J\) under catalog multiplicity.

Physical interpretation requires stable posterior mass on activation events, likelihood-level diagnostics that localize the relevant residuals, and quantitative checks of calibration and robustness. A claim of late-time background, ruler, or calibration support is physically interpretable when the grouped outputs \(\{p(\Act\mid D),P_\alpha,P_{\alpha\beta},\ln B_\alpha\}\) are consistent with the internal diagnostics \(\{\Delta\vect y,a_{\rd},\Delta\vect v_\perp,\widehat{\vect q},\vect r_\mu^\perp\}\), satisfy the predictive and robustness criteria, and survive controlled data-block replacement tests. Under those conditions the construction maps model-comparison support to a physically interpretable sector.

The construction also quantifies its principal specification dependence. Grouping removes within-pattern model multiplicity, while the declared sector partition remains part of the statistical model. If a broad perturbation sector is split into growth and lensing, or if an SN systematic requires two fine switches while a dark-energy parameterization requires one, a naive sparsity prior can become a physical tie-breaker. Section~\ref{subsec:sector_partition_prior} therefore requires pattern-neutral and sparsity-prior scans, coarsening-compatible priors, sector-partition push-forward tests, and the prior-sensitivity statistic \(\Gamma_\alpha\). When competing pattern evidences are nearly tied, posterior ranking is interpreted together with these diagnostics.

The smallest observational test of the construction is a common-catalog three-sector analysis of \(\Theta_{E(z)}\), \(\Theta_{\rd}\), and \(\Theta_{\rm SN}\) on
\begin{equation}
\mathcal D_1=\{\mathrm{DESI\ BAO}+\mathrm{CMB}+\mathrm{SN}_{A}\},
\qquad
\mathcal D_2=\{\mathrm{DESI\ BAO}+\mathrm{CMB}+\mathrm{SN}_{B}\},
\label{eq:minimal_real_data_runs}
\end{equation}
where \(\mathrm{SN}_{A}\) and \(\mathrm{SN}_{B}\) denote alternative supernova-calibration likelihoods analyzed with the same activation-pattern definition, pattern prior, detection basis, likelihood specification, and marginal-likelihood calculation. The primary output is the change in posterior mass among activation patterns, reported together with posterior-contour shifts inside fixed late-time model families.

For the comparison of \(\mathcal D_1\) and \(\mathcal D_2\), a late-time interpretation requires stable preference for patterns with \(A_{E(z)}=1\) across the alternative SN likelihoods, together with acceptable BAO- and SN-block PPCs. A ruler interpretation requires corresponding stability in patterns with \(A_{\rd}=1\), localization in \(a_{\rd}\), and stable
\(\Delta\vect y\) and \(\Delta\vect v_\perp\). A calibration-driven
interpretation requires the dominant posterior movement under
\(\mathrm{SN}_{A}\rightarrow\mathrm{SN}_{B}\) to occur in \(P_{\rm SN}\) and/or \(\ln B_{\rm SN}\), together with localization in \(M\widehat{\vect q}\), rather than in \(P_{E(z)}\) or \(P_{\rd}\). In all cases, posterior-predictive failure in the affected block precludes a localized sector interpretation.

The immediate analysis is finite: apply the pattern-labeled catalog and the sector-resolved likelihood specification to the data combination of Eq.~(\ref{eq:current_geometry_stack}), starting with the eight-pattern
\((E(z),\rd,{\rm SN})\) analysis of Sec.~\ref{subsec:geometry_application}
and then adding perturbations and GW propagation. The required reporting set is
\(\{p(\Act\mid D),P_\alpha,P_{\alpha\beta},\ln B_\alpha,{\rm PPC},\mathcal R_\alpha\}\), supplemented by the BAO and SN diagnostics for the geometric subproblem. Without these sector-resolved outputs and the corresponding
calibration, leakage, prior, sector-partition, and robustness diagnostics, claims of dynamical dark energy remain model-comparison statements rather than identified physical-sector interpretations.

\section*{Acknowledgments}
The work described here was carried out at the Jet Propulsion Laboratory, California Institute of Technology, Pasadena, California, under a contract with the National Aeronautics and Space Administration.  \textcopyright\ 2026. California Institute of Technology. Government sponsorship acknowledged.

\appendix

\section{Technical details}
\label{app:technical}

\subsection{Proof of Proposition~\ref{prop:grouped_invariance}}
\label{app:proofs}

\begin{proof}
For every pattern \(\Act\),
\begin{align}
\widetilde Z_{\Act}
&=
\sum_{\widetilde i\in\widetilde{\Cat}_{\Act}}
\widetilde\pi(\widetilde i\mid\Act)Z_{\widetilde i}
=
\sum_{i\in\Cat_{\Act}}
\sum_{\widetilde i\in r^{-1}(i)}
\widetilde\pi(\widetilde i\mid\Act)Z_{\widetilde i}
=
\sum_{i\in\Cat_{\Act}}
\pi(i\mid\Act)Z_i
=
Z_{\Act}.
\end{align}
Eqs.~(\ref{eq:patternposterior}), (\ref{eq:sectorodds}), and (\ref{eq:pairwise_activation}) depend on the catalog only through \(\{Z_{\Act},\pi(\Act)\}_{\Act\in\ActSpace}\). Therefore \(\widetilde p(\Act\mid D)=p(\Act\mid D)\), \(\widetilde P_\alpha=P_\alpha\), \(\widetilde P_{\alpha\beta}=P_{\alpha\beta}\), and \(\widetilde B_\alpha(D)=B_\alpha(D)\).
\end{proof}

\subsection{Benchmark classes and calibration diagnostics}
\label{app:benchmarks}

A minimal validation family should contain three classes of benchmark scenarios:
(i) {\it Null class.} Null realizations test whether the grouped estimator spuriously activates a sector when no anomaly is present.
(ii) {\it Single-sector class.}
Single-sector realizations test whether the estimator correctly localizes a background, ruler, supernova, perturbation, or GW-sector activation when only one sector is active.
(iii) {\it Mixed-sector class.} Mixed-sector realizations test whether the estimator can distinguish composite signals from apparent single-sector activations.

For a posterior threshold \(\tau\), the null false-activation rate (FAR) is
\begin{equation}
\FAR_\alpha(\tau)
\equiv
\Pr\!\left(P_\alpha>\tau \,\middle|\, \text{null benchmark}\right),
\label{eq:far}
\end{equation}
and a simple calibration diagnostic is the expected calibration error (ECE), 
\begin{equation}
\ECE_\alpha
=
\sum_{k=1}^{K}
\frac{n_k}{N}
\Big|
\frac{1}{n_k}\sum_{i\in B_k} P_{\alpha,i}
-
\frac{1}{n_k}\sum_{i\in B_k} y_{\alpha,i}
\Big|,
\label{eq:ece}
\end{equation}
where \(y_{\alpha,i}\in\{0,1\}\) denotes whether sector \(\alpha\) is truly active in realization \(i\). In addition, sector-level posterior-predictive checks should be reported in the likelihood blocks that dominate the inferred sector preference.

\subsection{Robustness diagnostics}
\label{app:robustness}

For prior-width robustness, let \(\pi_\alpha^{(\rho)}\) denote a prior family
obtained by rescaling the support width of the sector-\(\alpha\) amplitudes by a
factor \(\rho\) at fixed shape and center. We then define
\begin{equation}
\eta_\alpha(\rho)
\equiv
\frac{
\ln B_\alpha\big[\pi_\alpha^{(\rho)}\big]
-
\ln B_\alpha\big[\pi_\alpha^{(1)}\big]
}{
\ln \rho
},
\qquad
\rho\in\left\{\tfrac{1}{2},2,4\right\}.
\label{eq:eta_rho}
\end{equation}

For catalog-refinement robustness, define
\begin{equation}
\delta_\alpha^{\rm ref}
\equiv
\max_{\mathcal R}
\left|
P_\alpha^{\mathcal R}-P_\alpha^{\rm base}
\right|,
\qquad
\zeta_\alpha^{\rm ref}
\equiv
\max_{\mathcal R}
\left|
\ln B_\alpha^{\mathcal R}-\ln B_\alpha^{\rm base}
\right|,
\label{eq:catalog_robustness}
\end{equation}
where \(\mathcal R\) ranges over sector-preserving refinements of the catalog.

For basis robustness, let \(\mathcal B_1\) and \(\mathcal B_2\) denote two admissible detection bases. We then define
\begin{equation}
D_{\rm JS}^{\rm basis}(\mathcal B_1,\mathcal B_2)
\equiv
\JS\!\left[
p(\Act\mid D,\mathcal B_1),
p(\Act\mid D,\mathcal B_2)
\right].
\label{eq:js_basis}
\end{equation}
The sector-partition diagnostics \(D_{\rm JS}^{\rm part}\), \(\Delta_\alpha^\lambda\), and \(\Gamma_\alpha\) are defined in Eqs.~(\ref{eq:sector_partition_js}), (\ref{eq:lambda_sensitivity}), and (\ref{eq:Gamma_alpha}), respectively. They are listed in the robustness vector because a sector claim should not survive only as a consequence of a particular high-level partition of the theory space.

\section{Illustrative source families for future pattern-labeled catalogs}
\label{app:illustrative_catalog}

This appendix shows how broad source families may be partitioned into pattern-labeled catalog elements for an analysis using the grouped construction. Table~\ref{tab:modelspace} lists illustrative source families from which a pattern-labeled inference catalog may be constructed. In the grouped analysis, broad source families must be split into pattern-labeled catalog elements before evidences are aggregated according to Sec.~\ref{subsec:grouped_estimator}. A source family listed in Table~\ref{tab:modelspace} need not coincide with a single inference element. In the grouped analysis, any source family whose admissible parameter space can realize more than one activation pattern must be partitioned into pattern-labeled catalog elements before evidences are aggregated. The catalog construction therefore proceeds in two steps: first, identify the broad source family; second, split its admissible parameter region into sectors or subfamilies with a unique activation pattern \(\Act(i)\). Only after that partition is the grouped evidence
\(Z_{\Act}=\sum_{i\in\Cat_{\Act}} Z_i \pi(i\mid\Act)\) well defined.

\begin{table}[!tbp]
\caption{Illustrative source families for constructing a pattern-labeled DESI-era grouped catalog. The entries are examples of broad families that may realize distinct activation sectors; an actual analysis must split any family whose admissible parameter region spans more than one activation pattern.}
\label{tab:modelspace}
\tablefont
\tighttable
\begin{threeparttable}
\begin{tabular}{@{}llll@{}}
\toprule
\tcell{0.17\textwidth}{Model class} &
\tcell{0.32\textwidth}{Representative parameters} &
\tcell{0.18\textwidth}{Activated sectors} &
\tcell{0.26\textwidth}{Inferential role} \\
\midrule
\tcell{0.17\textwidth}{Baseline cosmology} &
\tcell{0.32\textwidth}{base six parameters; optional \(\Sigma m_\nu\), \(\Omega_k\)} &
\tcell{0.18\textwidth}{Base} &
\tcell{0.26\textwidth}{Null family against which all extensions are judged} \\
\tcell{0.17\textwidth}{Smooth late-time dark energy} &
\tcell{0.32\textwidth}{\(w\), \((w_0,w_a)\), limited nonparametric \(\rho_{\rm DE}(z)\) freedom} &
\tcell{0.18\textwidth}{Base + late-time shape} &
\tcell{0.26\textwidth}{Minimal homogeneous alternatives to \lcdm\ and robustness against parameterization bias} \\
\tcell{0.17\textwidth}{Early-time ruler sector / early dark energy (EDE)} &
\tcell{0.32\textwidth}{\(\delta\ln\rd\), \((f_{\rm EDE},\log_{10}z_c,\theta_i)\)} &
\tcell{0.18\textwidth}{Base + early-time ruler} &
\tcell{0.26\textwidth}{Tests whether anomalies originate before recombination} \\
\tcell{0.17\textwidth}{Interacting or clustering dark energy} &
\tcell{0.32\textwidth}{\(\xi\), \(w(a)\), \(c_s^2\)} &
\tcell{0.18\textwidth}{Base + late-time shape + perturbations} &
\tcell{0.26\textwidth}{Allows coupled or clustered dark-sector dynamics} \\
\tcell{0.17\textwidth}{Phenomenological modified gravity} &
\tcell{0.32\textwidth}{\(\mu(z,k),\Sigma(z,k),\eta(z,k)\)} &
\tcell{0.18\textwidth}{Base + perturbations} &
\tcell{0.26\textwidth}{Detection basis for departures from general relativity (GR) in growth and lensing} \\
\tcell{0.17\textwidth}{Effective field theory (EFT)/Horndeski / \(f(R)\)} &
\tcell{0.32\textwidth}{\((\alpha_K,\alpha_B,\alpha_M,\alpha_T)\), \(f_{R0}\), \(B_0\)} &
\tcell{0.18\textwidth}{Base + perturbations + possibly GW propagation} &
\tcell{0.26\textwidth}{Theory-level interpretation of perturbation anomalies} \\
\tcell{0.17\textwidth}{GW propagation / polarization / strong field} &
\tcell{0.33\textwidth}{\((\Xi_0,n)\), polarization fractions, quasi- normal-mode (QNM) or inspiral--merger-- ringdown (IMR) hyperparameters} &
\tcell{0.18\textwidth}{Base + GW sectors} &
\tcell{0.26\textwidth}{Tensor-sector discriminator and gravitational-wave consistency tests} \\
\bottomrule
\end{tabular}
\end{threeparttable}
\end{table}


\begin{thebibliography}{43}%
\makeatletter
\providecommand \@ifxundefined [1]{%
 \@ifx{#1\undefined}
}%
\providecommand \@ifnum [1]{%
 \ifnum #1\expandafter \@firstoftwo
 \else \expandafter \@secondoftwo
 \fi
}%
\providecommand \@ifx [1]{%
 \ifx #1\expandafter \@firstoftwo
 \else \expandafter \@secondoftwo
 \fi
}%
\providecommand \natexlab [1]{#1}%
\providecommand \enquote  [1]{``#1''}%
\providecommand \bibnamefont  [1]{#1}%
\providecommand \bibfnamefont [1]{#1}%
\providecommand \citenamefont [1]{#1}%
\providecommand \href@noop [0]{\@secondoftwo}%
\providecommand \href [0]{\begingroup \@sanitize@url \@href}%
\providecommand \@href[1]{\@@startlink{#1}\@@href}%
\providecommand \@@href[1]{\endgroup#1\@@endlink}%
\providecommand \@sanitize@url [0]{\catcode `\\12\catcode `\$12\catcode
  `\&12\catcode `\#12\catcode `\^12\catcode `\_12\catcode `\%12\relax}%
\providecommand \@@startlink[1]{}%
\providecommand \@@endlink[0]{}%
\providecommand \url  [0]{\begingroup\@sanitize@url \@url }%
\providecommand \@url [1]{\endgroup\@href {#1}{\urlprefix }}%
\providecommand \urlprefix  [0]{URL }%
\providecommand \Eprint [0]{\href }%
\providecommand \doibase [0]{https://doi.org/}%
\providecommand \selectlanguage [0]{\@gobble}%
\providecommand \bibinfo  [0]{\@secondoftwo}%
\providecommand \bibfield  [0]{\@secondoftwo}%
\providecommand \translation [1]{[#1]}%
\providecommand \BibitemOpen [0]{}%
\providecommand \bibitemStop [0]{}%
\providecommand \bibitemNoStop [0]{.\EOS\space}%
\providecommand \EOS [0]{\spacefactor3000\relax}%
\providecommand \BibitemShut  [1]{\csname bibitem#1\endcsname}%
\let\auto@bib@innerbib\@empty
\bibitem [{\citenamefont {{Chevallier}}\ and\ \citenamefont
  {{Polarski}}(2001)}]{Chevallier2001}%
  \BibitemOpen
  \bibfield  {author} {\bibinfo {author} {\bibfnamefont {M.}~\bibnamefont
  {{Chevallier}}}\ and\ \bibinfo {author} {\bibfnamefont {D.}~\bibnamefont
  {{Polarski}}},\ }\bibfield  {title} {\bibinfo {title} {{Accelerating
  Universes with Scaling Dark Matter}},\ }\href
  {https://doi.org/10.1142/S0218271801000822} {\bibfield  {journal} {\bibinfo
  {journal} {IJMPD}\ }\textbf {\bibinfo {volume} {10}},\ \bibinfo {pages} {213}
  (\bibinfo {year} {2001})}\BibitemShut {NoStop}%
\bibitem [{\citenamefont {{Linder}}(2003)}]{Linder2003}%
  \BibitemOpen
  \bibfield  {author} {\bibinfo {author} {\bibfnamefont {E.~V.}\ \bibnamefont
  {{Linder}}},\ }\bibfield  {title} {\bibinfo {title} {{Exploring the Expansion
  History of the Universe}},\ }\href
  {https://doi.org/10.1103/PhysRevLett.90.091301} {\bibfield  {journal}
  {\bibinfo  {journal} {Phys. Rev. Lett.}\ }\textbf {\bibinfo {volume} {90}},\
  \bibinfo {eid} {091301} (\bibinfo {year} {2003})}\BibitemShut {NoStop}%
\bibitem [{\citenamefont {{DESI
  Collaboration}}(2025{\natexlab{a}})}]{DESIDR2BAO2025}%
  \BibitemOpen
  \bibfield  {author} {\bibinfo {author} {\bibnamefont {{DESI
  Collaboration}}},\ }\bibfield  {title} {\bibinfo {title} {{DESI DR2 results.
  II. Measurements of baryon acoustic oscillations and cosmological
  constraints}},\ }\href {https://doi.org/10.1103/tr6y-kpc6} {\bibfield
  {journal} {\bibinfo  {journal} {Phys. Rev. D}\ }\textbf {\bibinfo {volume}
  {112}},\ \bibinfo {eid} {083515} (\bibinfo {year}
  {2025}{\natexlab{a}})}\BibitemShut {NoStop}%
\bibitem [{\citenamefont {{DESI
  Collaboration}}(2025{\natexlab{b}})}]{DESIDR2Extended2025}%
  \BibitemOpen
  \bibfield  {author} {\bibinfo {author} {\bibnamefont {{DESI
  Collaboration}}},\ }\bibfield  {title} {\bibinfo {title} {{Extended dark
  energy analysis using DESI DR2 BAO measurements}},\ }\href
  {https://doi.org/10.1103/w4c6-1r5j} {\bibfield  {journal} {\bibinfo
  {journal} {Phys. Rev. D}\ }\textbf {\bibinfo {volume} {112}},\ \bibinfo {eid}
  {083511} (\bibinfo {year} {2025}{\natexlab{b}})}\BibitemShut {NoStop}%
\bibitem [{\citenamefont {Ong}\ \emph {et~al.}(2025)\citenamefont {Ong},
  \citenamefont {Yallup},\ and\ \citenamefont
  {Handley}}]{OngYallupHandley2025}%
  \BibitemOpen
  \bibfield  {author} {\bibinfo {author} {\bibfnamefont {D.~D.~Y.}\
  \bibnamefont {Ong}}, \bibinfo {author} {\bibfnamefont {D.}~\bibnamefont
  {Yallup}},\ and\ \bibinfo {author} {\bibfnamefont {W.}~\bibnamefont
  {Handley}},\ }\bibfield  {title} {\bibinfo {title} {{A Bayesian Perspective
  on Evidence for Evolving Dark Energy}},\ }\href@noop {} {\bibfield  {journal}
  {\bibinfo  {journal} {arXiv e-prints}\ } (\bibinfo {year} {2025})},\ \Eprint
  {https://arxiv.org/abs/2511.10631} {arXiv:2511.10631 [astro-ph.CO]}
  \BibitemShut {NoStop}%
\bibitem [{\citenamefont {Ong}\ \emph {et~al.}(2026)\citenamefont {Ong},
  \citenamefont {Yallup},\ and\ \citenamefont
  {Handley}}]{OngYallupHandley2026}%
  \BibitemOpen
  \bibfield  {author} {\bibinfo {author} {\bibfnamefont {D.~D.~Y.}\
  \bibnamefont {Ong}}, \bibinfo {author} {\bibfnamefont {D.}~\bibnamefont
  {Yallup}},\ and\ \bibinfo {author} {\bibfnamefont {W.}~\bibnamefont
  {Handley}},\ }\bibfield  {title} {\bibinfo {title} {{The Bayesian View of
  {DESI} {DR2}: Evidence and Tension in a Combined Analysis with {CMB} and
  Supernovae Across Cosmological Models}},\ }\href@noop {} {\bibfield
  {journal} {\bibinfo  {journal} {arXiv e-prints}\ } (\bibinfo {year}
  {2026})},\ \Eprint {https://arxiv.org/abs/2603.05472} {arXiv:2603.05472
  [astro-ph.CO]} \BibitemShut {NoStop}%
\bibitem [{\citenamefont {Popovic}\ \emph {et~al.}(2026)\citenamefont
  {Popovic}, \citenamefont {Shah}, \citenamefont {Kenworthy} \emph
  {et~al.}}]{DESDovekie2025}%
  \BibitemOpen
  \bibfield  {author} {\bibinfo {author} {\bibfnamefont {B.}~\bibnamefont
  {Popovic}}, \bibinfo {author} {\bibfnamefont {P.}~\bibnamefont {Shah}},
  \bibinfo {author} {\bibfnamefont {W.~D.}\ \bibnamefont {Kenworthy}}, \emph
  {et~al.},\ }\bibfield  {title} {\bibinfo {title} {{The Dark Energy Survey
  Supernova Program: A Reanalysis of Cosmology Results and Evidence for
  Evolving Dark Energy with an Updated Type Ia Supernova Calibration}},\ }\href
  {https://doi.org/10.1093/mnras/stag632} {\bibfield  {journal} {\bibinfo
  {journal} {Mon. Not. R. Astron. Soc.}\ }\textbf {\bibinfo {volume} {548}},\
  \bibinfo {pages} {stag632} (\bibinfo {year} {2026})},\ \Eprint
  {https://arxiv.org/abs/2511.07517} {arXiv:2511.07517 [astro-ph.CO]}
  \BibitemShut {NoStop}%
\bibitem [{\citenamefont {{Huang}}\ \emph {et~al.}(2025)\citenamefont
  {{Huang}}, \citenamefont {{Cai}},\ and\ \citenamefont
  {{Wang}}}]{HuangCaiWang2025}%
  \BibitemOpen
  \bibfield  {author} {\bibinfo {author} {\bibfnamefont {L.}~\bibnamefont
  {{Huang}}}, \bibinfo {author} {\bibfnamefont {R.-G.}\ \bibnamefont {{Cai}}},\
  and\ \bibinfo {author} {\bibfnamefont {S.-J.}\ \bibnamefont {{Wang}}},\
  }\bibfield  {title} {\bibinfo {title} {{The DESI DR1/DR2 evidence for
  dynamical dark energy is biased by low-redshift supernovae}},\ }\href
  {https://doi.org/10.1007/s11433-025-2754-5} {\bibfield  {journal} {\bibinfo
  {journal} {Sci. China Phys. Mech. Astron.}\ }\textbf {\bibinfo {volume}
  {68}},\ \bibinfo {eid} {100413} (\bibinfo {year} {2025})}\BibitemShut
  {NoStop}%
\bibitem [{\citenamefont {Carlin}\ and\ \citenamefont
  {Chib}(1995)}]{CarlinChib1995}%
  \BibitemOpen
  \bibfield  {author} {\bibinfo {author} {\bibfnamefont {B.~P.}\ \bibnamefont
  {Carlin}}\ and\ \bibinfo {author} {\bibfnamefont {S.}~\bibnamefont {Chib}},\
  }\bibfield  {title} {\bibinfo {title} {{Bayesian Model Choice via Markov
  Chain Monte Carlo Methods}},\ }\href
  {https://doi.org/10.1111/j.2517-6161.1995.tb02042.x} {\bibfield  {journal}
  {\bibinfo  {journal} {J. R. Stat. Soc. Ser. B}\ }\textbf {\bibinfo {volume}
  {57}},\ \bibinfo {pages} {473} (\bibinfo {year} {1995})}\BibitemShut
  {NoStop}%
\bibitem [{\citenamefont {Green}(1995)}]{Green1995}%
  \BibitemOpen
  \bibfield  {author} {\bibinfo {author} {\bibfnamefont {P.~J.}\ \bibnamefont
  {Green}},\ }\bibfield  {title} {\bibinfo {title} {{Reversible Jump Markov
  Chain Monte Carlo Computation and Bayesian Model Determination}},\ }\href
  {https://doi.org/10.1093/biomet/82.4.711} {\bibfield  {journal} {\bibinfo
  {journal} {Biometrika}\ }\textbf {\bibinfo {volume} {82}},\ \bibinfo {pages}
  {711} (\bibinfo {year} {1995})}\BibitemShut {NoStop}%
\bibitem [{\citenamefont {Trotta}(2008)}]{Trotta2008}%
  \BibitemOpen
  \bibfield  {author} {\bibinfo {author} {\bibfnamefont {R.}~\bibnamefont
  {Trotta}},\ }\bibfield  {title} {\bibinfo {title} {{Bayes in the Sky:
  Bayesian Inference and Model Selection in Cosmology}},\ }\href
  {https://doi.org/10.1080/00107510802066753} {\bibfield  {journal} {\bibinfo
  {journal} {Contemp. Phys.}\ }\textbf {\bibinfo {volume} {49}},\ \bibinfo
  {pages} {71} (\bibinfo {year} {2008})}\BibitemShut {NoStop}%
\bibitem [{\citenamefont {Hoeting}\ \emph {et~al.}(1999)\citenamefont
  {Hoeting}, \citenamefont {Madigan}, \citenamefont {Raftery},\ and\
  \citenamefont {Volinsky}}]{Hoeting1999BMA}%
  \BibitemOpen
  \bibfield  {author} {\bibinfo {author} {\bibfnamefont {J.~A.}\ \bibnamefont
  {Hoeting}}, \bibinfo {author} {\bibfnamefont {D.}~\bibnamefont {Madigan}},
  \bibinfo {author} {\bibfnamefont {A.~E.}\ \bibnamefont {Raftery}},\ and\
  \bibinfo {author} {\bibfnamefont {C.~T.}\ \bibnamefont {Volinsky}},\
  }\bibfield  {title} {\bibinfo {title} {{Bayesian Model Averaging: A
  Tutorial}},\ }\href {https://doi.org/10.1214/ss/1009212519} {\bibfield
  {journal} {\bibinfo  {journal} {Statistical Science}\ }\textbf {\bibinfo
  {volume} {14}},\ \bibinfo {pages} {382} (\bibinfo {year} {1999})}\BibitemShut
  {NoStop}%
\bibitem [{\citenamefont {{Handley}}\ and\ \citenamefont
  {{Lemos}}(2019)}]{HandleyLemos2019}%
  \BibitemOpen
  \bibfield  {author} {\bibinfo {author} {\bibfnamefont {W.}~\bibnamefont
  {{Handley}}}\ and\ \bibinfo {author} {\bibfnamefont {P.}~\bibnamefont
  {{Lemos}}},\ }\bibfield  {title} {\bibinfo {title} {{Quantifying
  Dimensionality: Bayesian Cosmological Model Complexities}},\ }\href
  {https://doi.org/10.1103/PhysRevD.100.023512} {\bibfield  {journal} {\bibinfo
   {journal} {Phys. Rev. D}\ }\textbf {\bibinfo {volume} {100}},\ \bibinfo
  {eid} {023512} (\bibinfo {year} {2019})}\BibitemShut {NoStop}%
\bibitem [{\citenamefont {{Lemos}}\ \emph {et~al.}(2020)\citenamefont
  {{Lemos}}, \citenamefont {{K{\"o}hlinger}}, \citenamefont {{Handley}},
  \citenamefont {{Joachimi}}, \citenamefont {{Whiteway}},\ and\ \citenamefont
  {{Lahav}}}]{Lemos2019Suspiciousness}%
  \BibitemOpen
  \bibfield  {author} {\bibinfo {author} {\bibfnamefont {P.}~\bibnamefont
  {{Lemos}}}, \bibinfo {author} {\bibfnamefont {F.}~\bibnamefont
  {{K{\"o}hlinger}}}, \bibinfo {author} {\bibfnamefont {W.}~\bibnamefont
  {{Handley}}}, \bibinfo {author} {\bibfnamefont {B.}~\bibnamefont
  {{Joachimi}}}, \bibinfo {author} {\bibfnamefont {L.}~\bibnamefont
  {{Whiteway}}},\ and\ \bibinfo {author} {\bibfnamefont {O.}~\bibnamefont
  {{Lahav}}},\ }\bibfield  {title} {\bibinfo {title} {{Quantifying
  Suspiciousness within Correlated Data Sets}},\ }\href
  {https://doi.org/10.1093/mnras/staa1836} {\bibfield  {journal} {\bibinfo
  {journal} {MNRAS}\ }\textbf {\bibinfo {volume} {496}},\ \bibinfo {pages}
  {4647} (\bibinfo {year} {2020})}\BibitemShut {NoStop}%
\bibitem [{\citenamefont {Hergt}\ \emph {et~al.}(2026)\citenamefont {Hergt},
  \citenamefont {Henrot-Versill{\'e}}, \citenamefont {Tristram},\ and\
  \citenamefont {Scott}}]{Hergt2026Consistency}%
  \BibitemOpen
  \bibfield  {author} {\bibinfo {author} {\bibfnamefont {L.~T.}\ \bibnamefont
  {Hergt}}, \bibinfo {author} {\bibfnamefont {S.}~\bibnamefont
  {Henrot-Versill{\'e}}}, \bibinfo {author} {\bibfnamefont {M.}~\bibnamefont
  {Tristram}},\ and\ \bibinfo {author} {\bibfnamefont {D.}~\bibnamefont
  {Scott}},\ }\bibfield  {title} {\bibinfo {title} {{Consistency of Standard
  Cosmologies Using Bayesian Model Comparison and Tension Quantification}},\
  }\href@noop {} {\bibfield  {journal} {\bibinfo  {journal} {arXiv e-prints}\ }
  (\bibinfo {year} {2026})},\ \Eprint {https://arxiv.org/abs/2602.06115}
  {arXiv:2602.06115 [astro-ph.CO]} \BibitemShut {NoStop}%
\bibitem [{\citenamefont {Lewis}\ and\ \citenamefont
  {Challinor}(2006)}]{LewisChallinor2006}%
  \BibitemOpen
  \bibfield  {author} {\bibinfo {author} {\bibfnamefont {A.}~\bibnamefont
  {Lewis}}\ and\ \bibinfo {author} {\bibfnamefont {A.}~\bibnamefont
  {Challinor}},\ }\bibfield  {title} {\bibinfo {title} {{Weak Gravitational
  Lensing of the {CMB}}},\ }\href
  {https://doi.org/10.1016/j.physrep.2006.03.002} {\bibfield  {journal}
  {\bibinfo  {journal} {Phys. Rep.}\ }\textbf {\bibinfo {volume} {429}},\
  \bibinfo {pages} {1} (\bibinfo {year} {2006})},\ \Eprint
  {https://arxiv.org/abs/astro-ph/0601594} {arXiv:astro-ph/0601594}
  \BibitemShut {NoStop}%
\bibitem [{\citenamefont {{Planck
  Collaboration}}(2020{\natexlab{a}})}]{Planck2018Lensing}%
  \BibitemOpen
  \bibfield  {author} {\bibinfo {author} {\bibnamefont {{Planck
  Collaboration}}},\ }\bibfield  {title} {\bibinfo {title} {{Planck 2018
  Results. VIII. Gravitational Lensing}},\ }\href
  {https://doi.org/10.1051/0004-6361/201833886} {\bibfield  {journal} {\bibinfo
   {journal} {Astron. Astrophys.}\ }\textbf {\bibinfo {volume} {641}},\
  \bibinfo {pages} {A8} (\bibinfo {year} {2020}{\natexlab{a}})},\ \Eprint
  {https://arxiv.org/abs/1807.06210} {arXiv:1807.06210 [astro-ph.CO]}
  \BibitemShut {NoStop}%
\bibitem [{\citenamefont {Alcock}\ and\ \citenamefont
  {Paczy\'nski}(1979)}]{AlcockPaczynski1979}%
  \BibitemOpen
  \bibfield  {author} {\bibinfo {author} {\bibfnamefont {C.}~\bibnamefont
  {Alcock}}\ and\ \bibinfo {author} {\bibfnamefont {B.}~\bibnamefont
  {Paczy\'nski}},\ }\bibfield  {title} {\bibinfo {title} {{An Evolution Free
  Test for Non-zero Cosmological Constant}},\ }\href
  {https://doi.org/10.1038/281358a0} {\bibfield  {journal} {\bibinfo  {journal}
  {Nature}\ }\textbf {\bibinfo {volume} {281}},\ \bibinfo {pages} {358}
  (\bibinfo {year} {1979})}\BibitemShut {NoStop}%
\bibitem [{\citenamefont {Eisenstein}\ \emph {et~al.}(2005)\citenamefont
  {Eisenstein}, \citenamefont {Zehavi}, \citenamefont {Hogg}, \citenamefont
  {Scoccimarro}, \citenamefont {Blanton}, \citenamefont {Nichol}, \citenamefont
  {Scranton}, \citenamefont {Seo}, \citenamefont {Tegmark}, \citenamefont
  {Zheng} \emph {et~al.}}]{Eisenstein2005}%
  \BibitemOpen
  \bibfield  {author} {\bibinfo {author} {\bibfnamefont {D.~J.}\ \bibnamefont
  {Eisenstein}}, \bibinfo {author} {\bibfnamefont {I.}~\bibnamefont {Zehavi}},
  \bibinfo {author} {\bibfnamefont {D.~W.}\ \bibnamefont {Hogg}}, \bibinfo
  {author} {\bibfnamefont {R.}~\bibnamefont {Scoccimarro}}, \bibinfo {author}
  {\bibfnamefont {M.~R.}\ \bibnamefont {Blanton}}, \bibinfo {author}
  {\bibfnamefont {R.~C.}\ \bibnamefont {Nichol}}, \bibinfo {author}
  {\bibfnamefont {R.}~\bibnamefont {Scranton}}, \bibinfo {author}
  {\bibfnamefont {H.-J.}\ \bibnamefont {Seo}}, \bibinfo {author} {\bibfnamefont
  {M.}~\bibnamefont {Tegmark}}, \bibinfo {author} {\bibfnamefont
  {Z.}~\bibnamefont {Zheng}}, \emph {et~al.},\ }\bibfield  {title} {\bibinfo
  {title} {{Detection of the Baryon Acoustic Peak in the Large-Scale
  Correlation Function of {SDSS} Luminous Red Galaxies}},\ }\href
  {https://doi.org/10.1086/466512} {\bibfield  {journal} {\bibinfo  {journal}
  {Astrophys. J.}\ }\textbf {\bibinfo {volume} {633}},\ \bibinfo {pages} {560}
  (\bibinfo {year} {2005})},\ \Eprint {https://arxiv.org/abs/astro-ph/0501171}
  {arXiv:astro-ph/0501171} \BibitemShut {NoStop}%
\bibitem [{\citenamefont {Bernal}\ \emph {et~al.}(2020)\citenamefont {Bernal},
  \citenamefont {Smith}, \citenamefont {Boddy},\ and\ \citenamefont
  {Kamionkowski}}]{Bernal2020BAORobustness}%
  \BibitemOpen
  \bibfield  {author} {\bibinfo {author} {\bibfnamefont {J.~L.}\ \bibnamefont
  {Bernal}}, \bibinfo {author} {\bibfnamefont {T.~L.}\ \bibnamefont {Smith}},
  \bibinfo {author} {\bibfnamefont {K.~K.}\ \bibnamefont {Boddy}},\ and\
  \bibinfo {author} {\bibfnamefont {M.}~\bibnamefont {Kamionkowski}},\
  }\bibfield  {title} {\bibinfo {title} {{Robustness of Baryon Acoustic
  Oscillation Constraints for Early-Universe Modifications to $\Lambda$CDM}},\
  }\href {https://doi.org/10.1103/PhysRevD.102.123515} {\bibfield  {journal}
  {\bibinfo  {journal} {Phys. Rev. D}\ }\textbf {\bibinfo {volume} {102}},\
  \bibinfo {pages} {123515} (\bibinfo {year} {2020})}\BibitemShut {NoStop}%
\bibitem [{\citenamefont {Pan}\ \emph {et~al.}(2024)\citenamefont {Pan},
  \citenamefont {Huterer}, \citenamefont {Andrade-Oliveira},\ and\
  \citenamefont {Avestruz}}]{Pan2024CompressedBAO}%
  \BibitemOpen
  \bibfield  {author} {\bibinfo {author} {\bibfnamefont {J.}~\bibnamefont
  {Pan}}, \bibinfo {author} {\bibfnamefont {D.}~\bibnamefont {Huterer}},
  \bibinfo {author} {\bibfnamefont {F.}~\bibnamefont {Andrade-Oliveira}},\ and\
  \bibinfo {author} {\bibfnamefont {C.}~\bibnamefont {Avestruz}},\ }\bibfield
  {title} {\bibinfo {title} {{Compressed Baryon Acoustic Oscillation Analysis
  is Robust to Modified-Gravity Models}},\ }\href
  {https://doi.org/10.1088/1475-7516/2024/06/051} {\bibfield  {journal}
  {\bibinfo  {journal} {JCAP}\ }\textbf {\bibinfo {volume} {2024}}\bibinfo
  {number} { (06)},\ \bibinfo {pages} {051}}\BibitemShut {NoStop}%
\bibitem [{\citenamefont {Skilling}(2006)}]{Skilling2006Nested}%
  \BibitemOpen
\bibfield  {number} {  }\bibfield  {author} {\bibinfo {author} {\bibfnamefont
  {J.}~\bibnamefont {Skilling}},\ }\bibfield  {title} {\bibinfo {title}
  {{Nested Sampling for General Bayesian Computation}},\ }\href
  {https://doi.org/10.1214/06-BA127} {\bibfield  {journal} {\bibinfo  {journal}
  {Bayesian Analysis}\ }\textbf {\bibinfo {volume} {1}},\ \bibinfo {pages}
  {833} (\bibinfo {year} {2006})}\BibitemShut {NoStop}%
\bibitem [{\citenamefont {Feroz}\ \emph {et~al.}(2009)\citenamefont {Feroz},
  \citenamefont {Hobson},\ and\ \citenamefont {Bridges}}]{Feroz2009MultiNest}%
  \BibitemOpen
  \bibfield  {author} {\bibinfo {author} {\bibfnamefont {F.}~\bibnamefont
  {Feroz}}, \bibinfo {author} {\bibfnamefont {M.~P.}\ \bibnamefont {Hobson}},\
  and\ \bibinfo {author} {\bibfnamefont {M.}~\bibnamefont {Bridges}},\
  }\bibfield  {title} {\bibinfo {title} {{MultiNest: An Efficient and Robust
  Bayesian Inference Tool for Cosmology and Particle Physics}},\ }\href
  {https://doi.org/10.1111/j.1365-2966.2009.14548.x} {\bibfield  {journal}
  {\bibinfo  {journal} {Mon. Not. Roy. Astron. Soc.}\ }\textbf {\bibinfo
  {volume} {398}},\ \bibinfo {pages} {1601} (\bibinfo {year} {2009})},\ \Eprint
  {https://arxiv.org/abs/0809.3437} {arXiv:0809.3437 [astro-ph]} \BibitemShut
  {NoStop}%
\bibitem [{\citenamefont {Handley}\ \emph {et~al.}(2015)\citenamefont
  {Handley}, \citenamefont {Hobson},\ and\ \citenamefont
  {Lasenby}}]{Handley2015PolyChord}%
  \BibitemOpen
  \bibfield  {author} {\bibinfo {author} {\bibfnamefont {W.~J.}\ \bibnamefont
  {Handley}}, \bibinfo {author} {\bibfnamefont {M.~P.}\ \bibnamefont
  {Hobson}},\ and\ \bibinfo {author} {\bibfnamefont {A.~N.}\ \bibnamefont
  {Lasenby}},\ }\bibfield  {title} {\bibinfo {title} {{PolyChord: Nested
  Sampling for Cosmology}},\ }\href {https://doi.org/10.1093/mnrasl/slv047}
  {\bibfield  {journal} {\bibinfo  {journal} {Mon. Not. Roy. Astron. Soc.
  Lett.}\ }\textbf {\bibinfo {volume} {450}},\ \bibinfo {pages} {L61} (\bibinfo
  {year} {2015})},\ \Eprint {https://arxiv.org/abs/1502.01856}
  {arXiv:1502.01856 [astro-ph.CO]} \BibitemShut {NoStop}%
\bibitem [{\citenamefont {Speagle}(2020)}]{Speagle2020Dynesty}%
  \BibitemOpen
  \bibfield  {author} {\bibinfo {author} {\bibfnamefont {J.~S.}\ \bibnamefont
  {Speagle}},\ }\bibfield  {title} {\bibinfo {title} {{dynesty: A Dynamic
  Nested Sampling Package for Estimating Bayesian Posteriors and Evidences}},\
  }\href {https://doi.org/10.1093/mnras/staa278} {\bibfield  {journal}
  {\bibinfo  {journal} {Mon. Not. Roy. Astron. Soc.}\ }\textbf {\bibinfo
  {volume} {493}},\ \bibinfo {pages} {3132} (\bibinfo {year} {2020})},\ \Eprint
  {https://arxiv.org/abs/1904.02180} {arXiv:1904.02180 [astro-ph.IM]}
  \BibitemShut {NoStop}%
\bibitem [{\citenamefont {{Torrado}}\ and\ \citenamefont
  {{Lewis}}(2021)}]{TorradoLewis2021Cobaya}%
  \BibitemOpen
  \bibfield  {author} {\bibinfo {author} {\bibfnamefont {J.}~\bibnamefont
  {{Torrado}}}\ and\ \bibinfo {author} {\bibfnamefont {A.}~\bibnamefont
  {{Lewis}}},\ }\bibfield  {title} {\bibinfo {title} {{Cobaya: code for
  Bayesian analysis of hierarchical physical models}},\ }\href
  {https://doi.org/10.1088/1475-7516/2021/05/057} {\bibfield  {journal}
  {\bibinfo  {journal} {JCAP}\ }\textbf {\bibinfo {volume} {2021}}\bibinfo
  {number} { (5)},\ \bibinfo {eid} {057}}\BibitemShut {NoStop}%
\bibitem [{\citenamefont {{Lewis}}\ \emph {et~al.}(2000)\citenamefont
  {{Lewis}}, \citenamefont {{Challinor}},\ and\ \citenamefont
  {{Lasenby}}}]{Lewis2000CAMB}%
  \BibitemOpen
\bibfield  {number} {  }\bibfield  {author} {\bibinfo {author} {\bibfnamefont
  {A.}~\bibnamefont {{Lewis}}}, \bibinfo {author} {\bibfnamefont
  {A.}~\bibnamefont {{Challinor}}},\ and\ \bibinfo {author} {\bibfnamefont
  {A.}~\bibnamefont {{Lasenby}}},\ }\bibfield  {title} {\bibinfo {title}
  {{Efficient Computation of Cosmic Microwave Background Anisotropies in Closed
  Friedmann-Robertson-Walker Models}},\ }\href {https://doi.org/10.1086/309179}
  {\bibfield  {journal} {\bibinfo  {journal} {Astrophys. J.}\ }\textbf
  {\bibinfo {volume} {538}},\ \bibinfo {pages} {473} (\bibinfo {year}
  {2000})}\BibitemShut {NoStop}%
\bibitem [{\citenamefont {{Blas}}\ \emph {et~al.}(2011)\citenamefont {{Blas}},
  \citenamefont {{Lesgourgues}},\ and\ \citenamefont {{Tram}}}]{Blas2011CLASS}%
  \BibitemOpen
  \bibfield  {author} {\bibinfo {author} {\bibfnamefont {D.}~\bibnamefont
  {{Blas}}}, \bibinfo {author} {\bibfnamefont {J.}~\bibnamefont
  {{Lesgourgues}}},\ and\ \bibinfo {author} {\bibfnamefont {T.}~\bibnamefont
  {{Tram}}},\ }\bibfield  {title} {\bibinfo {title} {{The Cosmic Linear
  Anisotropy Solving System (CLASS). Part II: Approximation schemes}},\ }\href
  {https://doi.org/10.1088/1475-7516/2011/07/034} {\bibfield  {journal}
  {\bibinfo  {journal} {JCAP}\ }\textbf {\bibinfo {volume} {2011}}\bibinfo
  {number} { (7)},\ \bibinfo {eid} {034}}\BibitemShut {NoStop}%
\bibitem [{\citenamefont {{Hu}}\ \emph {et~al.}(2014)\citenamefont {{Hu}},
  \citenamefont {{Raveri}}, \citenamefont {{Frusciante}},\ and\ \citenamefont
  {{Silvestri}}}]{Hu2014EFTCAMB}%
  \BibitemOpen
\bibfield  {number} {  }\bibfield  {author} {\bibinfo {author} {\bibfnamefont
  {B.}~\bibnamefont {{Hu}}}, \bibinfo {author} {\bibfnamefont {M.}~\bibnamefont
  {{Raveri}}}, \bibinfo {author} {\bibfnamefont {N.}~\bibnamefont
  {{Frusciante}}},\ and\ \bibinfo {author} {\bibfnamefont {A.}~\bibnamefont
  {{Silvestri}}},\ }\bibfield  {title} {\bibinfo {title} {{Effective field
  theory of cosmic acceleration: An implementation in CAMB}},\ }\href
  {https://doi.org/10.1103/PhysRevD.89.103530} {\bibfield  {journal} {\bibinfo
  {journal} {Phys. Rev. D}\ }\textbf {\bibinfo {volume} {89}},\ \bibinfo {eid}
  {103530} (\bibinfo {year} {2014})}\BibitemShut {NoStop}%
\bibitem [{\citenamefont {{Zumalac{\'a}rregui}}\ \emph
  {et~al.}(2017)\citenamefont {{Zumalac{\'a}rregui}}, \citenamefont
  {{Bellini}}, \citenamefont {{Sawicki}}, \citenamefont {{Lesgourgues}},\ and\
  \citenamefont {{Ferreira}}}]{Zumalacarregui2017HICLASS}%
  \BibitemOpen
  \bibfield  {author} {\bibinfo {author} {\bibfnamefont {M.}~\bibnamefont
  {{Zumalac{\'a}rregui}}}, \bibinfo {author} {\bibfnamefont {E.}~\bibnamefont
  {{Bellini}}}, \bibinfo {author} {\bibfnamefont {I.}~\bibnamefont
  {{Sawicki}}}, \bibinfo {author} {\bibfnamefont {J.}~\bibnamefont
  {{Lesgourgues}}},\ and\ \bibinfo {author} {\bibfnamefont {P.~G.}\
  \bibnamefont {{Ferreira}}},\ }\bibfield  {title} {\bibinfo {title}
  {{hi\_class: Horndeski in the Cosmic Linear Anisotropy Solving System}},\
  }\href {https://doi.org/10.1088/1475-7516/2017/08/019} {\bibfield  {journal}
  {\bibinfo  {journal} {JCAP}\ }\textbf {\bibinfo {volume} {2017}}\bibinfo
  {number} { (8)},\ \bibinfo {eid} {019}}\BibitemShut {NoStop}%
\bibitem [{\citenamefont {{DESI
  Collaboration}}(2025{\natexlab{c}})}]{DESIDR2Lya2025}%
  \BibitemOpen
\bibfield  {number} {  }\bibfield  {author} {\bibinfo {author} {\bibnamefont
  {{DESI Collaboration}}},\ }\bibfield  {title} {\bibinfo {title} {{DESI DR2
  results. I. Baryon acoustic oscillations from the Lyman alpha forest}},\
  }\href {https://doi.org/10.1103/2wwn-xjm5} {\bibfield  {journal} {\bibinfo
  {journal} {Phys. Rev. D}\ }\textbf {\bibinfo {volume} {112}},\ \bibinfo {eid}
  {083514} (\bibinfo {year} {2025}{\natexlab{c}})}\BibitemShut {NoStop}%
\bibitem [{\citenamefont {{Planck
  Collaboration}}(2020{\natexlab{b}})}]{Planck2018Parameters}%
  \BibitemOpen
  \bibfield  {author} {\bibinfo {author} {\bibnamefont {{Planck
  Collaboration}}},\ }\bibfield  {title} {\bibinfo {title} {{Planck 2018
  Results. VI. Cosmological Parameters}},\ }\href
  {https://doi.org/10.1051/0004-6361/201833910} {\bibfield  {journal} {\bibinfo
   {journal} {Astron. Astrophys.}\ }\textbf {\bibinfo {volume} {641}},\
  \bibinfo {pages} {A6} (\bibinfo {year} {2020}{\natexlab{b}})},\ \Eprint
  {https://arxiv.org/abs/1807.06209} {arXiv:1807.06209 [astro-ph.CO]}
  \BibitemShut {NoStop}%
\bibitem [{\citenamefont {{ACT Collaboration}}(2025)}]{ACTDR6Extended2025}%
  \BibitemOpen
  \bibfield  {author} {\bibinfo {author} {\bibnamefont {{ACT Collaboration}}},\
  }\bibfield  {title} {\bibinfo {title} {{The Atacama Cosmology Telescope: DR6
  Constraints on Extended Cosmological Models}},\ }\href
  {https://doi.org/10.1088/1475-7516/2025/11/063} {\bibfield  {journal}
  {\bibinfo  {journal} {JCAP}\ }\textbf {\bibinfo {volume} {2025}}\bibfield
  {number} {\bibinfo  {number} { (11)},\ \bibinfo {pages} {063}},\ }\Eprint
  {https://arxiv.org/abs/2503.14454} {arXiv:2503.14454 [astro-ph.CO]}
  \BibitemShut {NoStop}%
\bibitem [{\citenamefont {Brout}\ \emph {et~al.}(2022)\citenamefont {Brout},
  \citenamefont {Scolnic}, \citenamefont {Popovic}, \citenamefont {Riess},
  \citenamefont {Zuntz}, \citenamefont {Kessler}, \citenamefont {Davis},
  \citenamefont {Hinton}, \citenamefont {Jones} \emph
  {et~al.}}]{BroutPantheonPlus2022}%
  \BibitemOpen
  \bibfield  {author} {\bibinfo {author} {\bibfnamefont {D.}~\bibnamefont
  {Brout}}, \bibinfo {author} {\bibfnamefont {D.}~\bibnamefont {Scolnic}},
  \bibinfo {author} {\bibfnamefont {B.}~\bibnamefont {Popovic}}, \bibinfo
  {author} {\bibfnamefont {A.~G.}\ \bibnamefont {Riess}}, \bibinfo {author}
  {\bibfnamefont {J.}~\bibnamefont {Zuntz}}, \bibinfo {author} {\bibfnamefont
  {R.}~\bibnamefont {Kessler}}, \bibinfo {author} {\bibfnamefont {T.~M.}\
  \bibnamefont {Davis}}, \bibinfo {author} {\bibfnamefont {S.}~\bibnamefont
  {Hinton}}, \bibinfo {author} {\bibfnamefont {D.}~\bibnamefont {Jones}}, \emph
  {et~al.},\ }\bibfield  {title} {\bibinfo {title} {{The Pantheon+ Analysis:
  Cosmological Constraints}},\ }\href
  {https://doi.org/10.3847/1538-4357/ac8e04} {\bibfield  {journal} {\bibinfo
  {journal} {Astrophys. J.}\ }\textbf {\bibinfo {volume} {938}},\ \bibinfo
  {pages} {110} (\bibinfo {year} {2022})},\ \Eprint
  {https://arxiv.org/abs/2202.04077} {arXiv:2202.04077 [astro-ph.CO]}
  \BibitemShut {NoStop}%
\bibitem [{\citenamefont {Rubin}\ \emph {et~al.}(2025)\citenamefont {Rubin},
  \citenamefont {Aldering}, \citenamefont {Betoule}, \citenamefont {Fruchter},
  \citenamefont {Huang} \emph {et~al.}}]{RubinUnion32023}%
  \BibitemOpen
  \bibfield  {author} {\bibinfo {author} {\bibfnamefont {D.}~\bibnamefont
  {Rubin}}, \bibinfo {author} {\bibfnamefont {G.}~\bibnamefont {Aldering}},
  \bibinfo {author} {\bibfnamefont {M.}~\bibnamefont {Betoule}}, \bibinfo
  {author} {\bibfnamefont {A.}~\bibnamefont {Fruchter}}, \bibinfo {author}
  {\bibfnamefont {X.}~\bibnamefont {Huang}}, \emph {et~al.},\ }\bibfield
  {title} {\bibinfo {title} {{Union Through UNITY: Cosmology with 2000 SNe
  Using a Unified Bayesian Framework}},\ }\href
  {https://doi.org/10.3847/1538-4357/adc0a5} {\bibfield  {journal} {\bibinfo
  {journal} {Astrophys. J.}\ }\textbf {\bibinfo {volume} {986}},\ \bibinfo
  {pages} {231} (\bibinfo {year} {2025})},\ \Eprint
  {https://arxiv.org/abs/2311.12098} {arXiv:2311.12098 [astro-ph.CO]}
  \BibitemShut {NoStop}%
\bibitem [{\citenamefont {{DES Collaboration}}(2024)}]{DESY5SN2024}%
  \BibitemOpen
  \bibfield  {author} {\bibinfo {author} {\bibnamefont {{DES Collaboration}}},\
  }\bibfield  {title} {\bibinfo {title} {{The Dark Energy Survey: Cosmology
  Results with $\sim$1500 New High-redshift Type Ia Supernovae Using the Full 5
  yr Data Set}},\ }\href {https://doi.org/10.3847/2041-8213/ad6f9f} {\bibfield
  {journal} {\bibinfo  {journal} {Astrophys. J. Lett.}\ }\textbf {\bibinfo
  {volume} {973}},\ \bibinfo {pages} {L14} (\bibinfo {year} {2024})},\ \Eprint
  {https://arxiv.org/abs/2401.02929} {arXiv:2401.02929 [astro-ph.CO]}
  \BibitemShut {NoStop}%
\bibitem [{\citenamefont {Gelman}\ \emph {et~al.}(1996)\citenamefont {Gelman},
  \citenamefont {Meng},\ and\ \citenamefont {Stern}}]{GelmanMengStern1996PPC}%
  \BibitemOpen
  \bibfield  {author} {\bibinfo {author} {\bibfnamefont {A.}~\bibnamefont
  {Gelman}}, \bibinfo {author} {\bibfnamefont {X.-L.}\ \bibnamefont {Meng}},\
  and\ \bibinfo {author} {\bibfnamefont {H.}~\bibnamefont {Stern}},\ }\bibfield
   {title} {\bibinfo {title} {{Posterior Predictive Assessment of Model Fitness
  via Realized Discrepancies}},\ }\href@noop {} {\bibfield  {journal} {\bibinfo
   {journal} {Statistica Sinica}\ }\textbf {\bibinfo {volume} {6}},\ \bibinfo
  {pages} {733} (\bibinfo {year} {1996})}\BibitemShut {NoStop}%
\bibitem [{\citenamefont {Gneiting}\ and\ \citenamefont
  {Raftery}(2007)}]{GneitingRaftery2007}%
  \BibitemOpen
  \bibfield  {author} {\bibinfo {author} {\bibfnamefont {T.}~\bibnamefont
  {Gneiting}}\ and\ \bibinfo {author} {\bibfnamefont {A.~E.}\ \bibnamefont
  {Raftery}},\ }\bibfield  {title} {\bibinfo {title} {{Strictly Proper Scoring
  Rules, Prediction, and Estimation}},\ }\href
  {https://doi.org/10.1198/016214506000001437} {\bibfield  {journal} {\bibinfo
  {journal} {Journal of the American Statistical Association}\ }\textbf
  {\bibinfo {volume} {102}},\ \bibinfo {pages} {359} (\bibinfo {year}
  {2007})}\BibitemShut {NoStop}%
\bibitem [{\citenamefont {Clopper}\ and\ \citenamefont
  {Pearson}(1934)}]{ClopperPearson1934}%
  \BibitemOpen
  \bibfield  {author} {\bibinfo {author} {\bibfnamefont {C.~J.}\ \bibnamefont
  {Clopper}}\ and\ \bibinfo {author} {\bibfnamefont {E.~S.}\ \bibnamefont
  {Pearson}},\ }\bibfield  {title} {\bibinfo {title} {{The Use of Confidence or
  Fiducial Limits Illustrated in the Case of the Binomial}},\ }\href
  {https://doi.org/10.1093/biomet/26.4.404} {\bibfield  {journal} {\bibinfo
  {journal} {Biometrika}\ }\textbf {\bibinfo {volume} {26}},\ \bibinfo {pages}
  {404} (\bibinfo {year} {1934})}\BibitemShut {NoStop}%
\bibitem [{\citenamefont {Bellini}\ and\ \citenamefont
  {Sawicki}(2014)}]{BelliniSawicki2014}%
  \BibitemOpen
  \bibfield  {author} {\bibinfo {author} {\bibfnamefont {E.}~\bibnamefont
  {Bellini}}\ and\ \bibinfo {author} {\bibfnamefont {I.}~\bibnamefont
  {Sawicki}},\ }\bibfield  {title} {\bibinfo {title} {{Maximal Freedom at
  Minimum Cost: Linear Large-Scale Structure in General Modifications of
  Gravity}},\ }\href {https://doi.org/10.1088/1475-7516/2014/07/050} {\bibfield
   {journal} {\bibinfo  {journal} {JCAP}\ }\textbf {\bibinfo {volume}
  {2014}}\bibinfo  {number} { (07)},\ \bibinfo {pages} {050}}\BibitemShut
  {NoStop}%
\bibitem [{\citenamefont {{Ishak}}\ \emph {et~al.}(2025)\citenamefont
  {{Ishak}}, \citenamefont {{Pan}}, \citenamefont {{Calderon}} \emph
  {et~al.}}]{Ishak2025MGDESI}%
  \BibitemOpen
\bibfield  {number} {  }\bibfield  {author} {\bibinfo {author} {\bibfnamefont
  {M.}~\bibnamefont {{Ishak}}}, \bibinfo {author} {\bibfnamefont
  {J.}~\bibnamefont {{Pan}}}, \bibinfo {author} {\bibfnamefont
  {R.}~\bibnamefont {{Calderon}}}, \emph {et~al.},\ }\bibfield  {title}
  {\bibinfo {title} {{Modified gravity constraints from the full shape modeling
  of clustering measurements from DESI 2024}},\ }\href
  {https://doi.org/10.1088/1475-7516/2025/09/053} {\bibfield  {journal}
  {\bibinfo  {journal} {JCAP}\ }\textbf {\bibinfo {volume} {2025}}\bibinfo
  {number} { (9)},\ \bibinfo {eid} {053}}\BibitemShut {NoStop}%
\bibitem [{\citenamefont {Belgacem}\ \emph {et~al.}(2018)\citenamefont
  {Belgacem}, \citenamefont {Dirian}, \citenamefont {Foffa},\ and\
  \citenamefont {Maggiore}}]{Belgacem2018}%
  \BibitemOpen
\bibfield  {number} {  }\bibfield  {author} {\bibinfo {author} {\bibfnamefont
  {E.}~\bibnamefont {Belgacem}}, \bibinfo {author} {\bibfnamefont
  {Y.}~\bibnamefont {Dirian}}, \bibinfo {author} {\bibfnamefont
  {S.}~\bibnamefont {Foffa}},\ and\ \bibinfo {author} {\bibfnamefont
  {M.}~\bibnamefont {Maggiore}},\ }\bibfield  {title} {\bibinfo {title}
  {{Modified gravitational-wave propagation and standard sirens}},\ }\href
  {https://doi.org/10.1103/PhysRevD.98.023510} {\bibfield  {journal} {\bibinfo
  {journal} {Phys. Rev. D}\ }\textbf {\bibinfo {volume} {98}},\ \bibinfo
  {pages} {023510} (\bibinfo {year} {2018})}\BibitemShut {NoStop}%
\bibitem [{\citenamefont {{LIGO Scientific Collaboration}}\ \emph
  {et~al.}(2026)\citenamefont {{LIGO Scientific Collaboration}}, \citenamefont
  {{Virgo Collaboration}},\ and\ \citenamefont {{KAGRA
  Collaboration}}}]{LVK2026GWTC5}%
  \BibitemOpen
  \bibfield  {author} {\bibinfo {author} {\bibnamefont {{LIGO Scientific
  Collaboration}}}, \bibinfo {author} {\bibnamefont {{Virgo Collaboration}}},\
  and\ \bibinfo {author} {\bibnamefont {{KAGRA Collaboration}}},\ }\bibfield
  {title} {\bibinfo {title} {{{GWTC-5.0}: Constraints on the Cosmic Expansion
  Rate and Modified Gravitational-Wave Propagation}},\ }\href@noop {}
  {\bibfield  {journal} {\bibinfo  {journal} {arXiv e-prints}\ } (\bibinfo
  {year} {2026})},\ \Eprint {https://arxiv.org/abs/2605.27227}
  {arXiv:2605.27227 [gr-qc]} \BibitemShut {NoStop}%
\end{thebibliography}

%

\end{document}